\documentclass[11pt,oneside]{article}

\usepackage[margin=1in]{geometry}
\usepackage{graphicx} 
\usepackage{subcaption}
\usepackage{xspace}
\usepackage{amssymb}
\usepackage{amsmath}
\usepackage{amsthm}
\usepackage{graphicx}
\usepackage{multicol}
\usepackage{diagbox}
\usepackage[colorlinks=true, allcolors=blue]{hyperref}
\usepackage{enumerate}    
\usepackage{algorithm}
\usepackage{algorithmic}
\usepackage{enumitem}
\usepackage{tikz}
\usetikzlibrary{positioning,automata}
\usetikzlibrary{decorations.pathreplacing,fit}
\usepackage{pifont}
\usepackage{booktabs}
\usepackage[switch]{lineno}

\newcommand{\problem}[1]{\textsc{#1}}
\newtheorem{theorem}{Theorem}[section]
\newtheorem{lemma}[theorem]{Lemma}
\newtheorem{corollary}[theorem]{Corollary}
\newtheorem{claim}[theorem]{Claim}

\newtheorem{definition}[theorem]{Definition}
\newtheorem{observation}[theorem]{Observation}
\newtheorem{example}[theorem]{Example}

\title{\textbf{The Complexity of Minimizing Subsidies in Envy-Free House Allocation}}
\author{\vspace{1em}
Sijia Dai$^{1,2}$ \quad
Minming Li$^{1}$ \quad
Xiaowei Wu$^{3}$ \quad
Yong Zhang$^{2}$
}
\date{} 

\begin{document}

\maketitle
\footnotetext[1]{City University of Hong Kong. Emails: 
\texttt{sijiadai2-c@my.cityu.edu.hk} and \texttt{minming.li@cityu.edu.hk}.}
\footnotetext[2]{Shenzhen Institutes of Advanced Technology, Chinese Academy of Sciences. Email:  \texttt{zhangyong@siat.ac.cn}.}
\footnotetext[3]{University of Macau. Email: \texttt{xiaoweiwu@um.edu.mo}.}
\begingroup
\renewcommand{\thefootnote}{}
\endgroup
\begin{abstract}

The house allocation problem is a classical one-sided matching problem that concerns the assignment of a set of $m$ houses to $n$ agents according to their preferences, where each agent is assigned exactly one house. Among the various objectives studied in this setting, envy-freeness is one of the most widely adopted fairness criteria. As envy-free house allocations do not always exist, we address this challenge by introducing subsidies and aim to compute allocations that achieve envy-freeness with minimum total subsidy. For binary instances, we show that a total subsidy of at most $(n-1)$ suffices to guarantee envy-freeness in house allocation, and this bound is tight. 
Building on the known NP-hardness for general utilities, we further show that computing an allocation that minimizes the total subsidy is NP-hard, even under binary utilities. However, when there are only a bounded number of types of agents with binary utilities, the problem can be solved in polynomial time. Finally, we present a polynomial time algorithm that computes the minimum subsidy required to achieve envy-freeness for two types of agents with general utilities.

\end{abstract}

\section{Introduction}

House allocation is a fundamental problem in the study of resource allocation. In this setting, a collection of indivisible items, such as houses, offices, or tasks, must be assigned to a group of agents with potentially different preferences, where each agent is allocated exactly one item~\cite{zhou1990,abdulkadirouglu1998}.
As a fundamental class of one-sided matching problems, house allocation has become increasingly critical due to the widespread adoption of automated decision-making systems. 
In many practical applications, houses are assigned through centralized procedures rather than market transactions. Examples include student dormitory assignments, employee housing provided by institutions, and public housing programs. Since such resources are often allocated according to administrative rules instead of prices, the quality of an allocation depends largely on the underlying mechanism. This has motivated extensive research on allocation rules that provide fairness and efficiency guarantees.

Envy-freeness (EF) is a central fairness notion in house allocation~\cite{beynier2019local,Haris2024}. An allocation satisfies envy-freeness if no agent prefers the house assigned to another agent over her own. However, an envy-free allocation may not exist for every instance. 
The question of whether an envy-free allocation exists was completely resolved by Gan et al.~\cite{gan2019envy}, who presented a polynomial-time algorithm for determining whether an envy-free allocation exists and for computing one whenever it does.
Since such allocations are not guaranteed to exist, minimizing unavoidable envy becomes a critical objective.
Kamiyama et al.~\cite{kamiyama2021complexity} demonstrated that minimizing the number of envious agents is NP-hard, even when agents have binary valuations. 
To address this intractability, Hosseini et al.~\cite{hosseini2025fair} proposed a fixed-parameter tractable framework for refining allocations via bounded reallocations, as well as polynomial-time algorithms for single-peaked preference. 
Complementing this, Madathil et al.~\cite{madathil2025cost} proposed three specific metrics for quantifying envy, confirming that minimizing these objectives generally remains computationally hard.

To circumvent the potential non-existence of envy-free allocations, recent works have turned to monetary subsidies, which have been extensively studied in fair division~\cite{SAGT2019,conf/sigecom/BrustleDNSV20,kawase2025towards}. 
While Gan et al.~\cite{gan2019envy} showed that the existence of an envy-free allocation with zero subsidy can be decided in polynomial time, Halpern and Shah~\cite{SAGT2019} established that envy-freeness can always be achieved if subsidies are allowed. 
However, Choo et al.~\cite{CHOO2024} demonstrated that finding the minimum subsidy required is NP-hard in general cases. These findings naturally raise the following question:
\begin{quote}
\emph{Can we efficiently determine the minimum subsidy required to achieve envy-freeness when this subsidy is non-zero for some restricted instances?}
\end{quote}

Motivated by this, we investigate the computational complexity of minimizing total subsidy in house allocation under several well-studied restricted settings, including binary instances, instances with a bounded number of agent types, and the specific case of two agent types.

\subsection{Our Contributions}

Our work considers the problem of minimizing total subsidies under the fairness notion of envy-freeness.
We present the results regarding the computational complexity of envy-free allocations with minimum subsidies.

\noindent
\paragraph{NP-hardness.} For binary instances, we represent house allocation instances as bipartite graphs and present the corresponding structural properties.
We refer to our problem of interest as \problem{Minimum Subsidy for Binary}.
We observe that Pareto optimality is compatible with minimizing the total subsidy in this setting (Lemma~\ref{subsidy-is-0}).
However, we show that finding a Pareto optimal allocation that minimizes the total subsidy is NP-hard even under binary utilities via a reduction from \textsc{Minimum $k$-Union} problem (Theorem~\ref{theorem:reduction}).
Furthermore, the reduction can be modified to be approximation-preserving, which directly yields an $\Omega(n^{1/4})$-approximation hardness for \problem{Minimum Subsidy for Binary} (see Appendix~\ref{Appendix:approximation-preserving}).

\paragraph{Fixed Parameter Tractability.} Given a bounded number $k$ of types of agents with binary utilities, we design a fixed-parameter tractable (FPT)\footnote{An FPT algorithm with respect to a parameter $k$ runs in time $f(k)(n + m)^{O(1)}$ for a computable function $f$.} algorithm for computing an allocation that minimizes the total subsidy (Theorem~\ref{theorem:k-types}).
The key structural insight is that an optimal allocation can be characterized in terms of agent types. In particular, agents of the same type cannot play different roles in the underlying subsidy structure.
This reduces the search to a bounded number of configurations over agent types, each of which can be handled via capacitated matching.
Therefore, the problem is fixed-parameter tractable when parameterized by the number of agent types.

\paragraph{Polynomial Time Algorithm.} We also consider the setting of two types of agents with general utilities, which is another widely studied setting in fair allocation~\cite{mahara2023existence,garg2024weighted,ijcai2025p452,ma2025maximizing}.
While minimizing the total subsidy is NP-hard in general~\cite{CHOO2024}, we exploit the structural properties of instances with two types of agents to characterize minimum-subsidy allocations and identify the properties that every such allocation must satisfy.
Based on these characterizations, we present a polynomial-time algorithm that computes an allocation with the minimum total subsidy (Theorem~\ref{theorem:type-min-subsidy-EF}).

\subsection{Other Related Works}

In fair division, deciding the existence of envy-free allocations with zero subsidy is already NP-hard. 
Consequently, the use of subsidies has been extensively explored to achieve fairness~\cite{conf/wine/WuZZ23,elmalem2025whoever,wu2025little,wu2025revisiting,li2025subsidy,garg2025proportional}.
A related domain is the rent division problem, which involves allocating $n$ distinct goods to $n$ agents while splitting the total rent~\cite{svensson1983large,edward1999rental,gal2017fairest}. 
Further research explored broader classes of valuations beyond additive ones, such as general monotone functions~\cite{kawase2025towards} and matroid  functions~\cite{barman2022achieving,goko2024fair}.
Other related works explored weighted fairness notions with subsidies~\cite{conf/wine/WuZZ23,wu2024tree,elmalem2025whoever}.

Recent years have seen renewed interest in the house allocation problem, driven by new challenges in fairness and computation. 
Shende and Purohit~\cite{shende2023strategy} studied the trade-offs between envy-freeness and strategy-proofness, 
while Hosseini et al.~\cite{hosseini2024degree} investigated the trade-off between efficiency and fairness in house allocation problems.
Dai et al.~\cite{dai2024weighted} introduced weighted envy-freeness in house allocation and investigated the use of subsidies to achieve it.
Aziz et al.~\cite{Haris2024} extended house allocation to settings with uncertain preferences.
Other works studied scenarios where agents envy others only along a predefined social graph, aiming to minimize total envy~\cite{hosseini2023graphical,hosseini2024tight}. 
For broader overviews of fair division of indivisible goods, we refer to the surveys~\cite{amanatidis2023fair,liu2024mixed,suksompong2024weighted}.

\section{Preliminaries}
\label{section:Preliminaries}

In this section, we introduce the notation and fairness concepts used throughout the paper.
Let $N=\{1,2,\ldots,n\}$ be the set of $n$ agents and $H=\{h_1,\ldots,h_m\}$ be the set of $m$ houses, where $m\ge n$. 
Each agent $i\in N$ is associated with a non-negative utility function $v_i:H\rightarrow\mathbb{R}_{\ge 0}$. 
For every house $h\in H$, $v_i(h)$ denotes the utility of agent $i$ for house $h$.
We use $\mathbf{v}=(v_1,\ldots,v_n)$ to denote the collection of all utility functions.
An allocation is represented by $\mathbf{A}=(A_1,\ldots,A_n)$, where $A_i\in H$ is the house assigned to agent $i$, and each house is assigned to at most one agent. 
Equivalently, an allocation can be viewed as an $N$-saturating matching\footnote{A matching is called $N$-saturating if every agent in $N$ is matched.} in the complete bipartite graph with bipartition $(N,H)$.
A house allocation instance is specified by the tuple $\mathcal{I}=(N,H,\mathbf{v})$. 
We next introduce the fairness notion that serves as the foundation of our study.

\begin{definition}[EF Allocation]
\label{definition:EF Allocation}
An allocation $\mathbf{A}$ is called envy-free (EF) if, for every pair of agents $i,j\in N$, 
\begin{equation*}
    v_i(A_i) \ge v_i(A_j).
\end{equation*}
\end{definition}

Since envy-free allocations may not exist, we consider the use of subsidies to achieve envy-freeness.
Let $p_i\in\mathbb{R}_{\ge 0}$ denote the subsidy received by agent $i$, and let $\mathbf{P}=(p_1,p_2,\ldots,p_n)$ denote the corresponding \emph{subsidy vector}. 
An \emph{outcome} is a pair $(\mathbf{A},\mathbf{P})$, consisting of an allocation $\mathbf{A}$ together with a corresponding subsidy vector $\mathbf{P}$. 
The utility of agent $i$ under an outcome $(\mathbf{A},\mathbf{P})$ is defined as $v_i(A_i)+p_i$. 
Note that an allocation $\mathbf{A}$ can be viewed as the special outcome $(\mathbf{A},\mathbf{0})$, where $\mathbf{0}$ denotes the all-zero subsidy vector. 
We next define the notion of an envy-free outcome.

\begin{definition}[EF Outcome]
\label{definition:EF Outcome}
An outcome $(\mathbf{A},\mathbf{P})$ is called envy-free if, for every pair of agents $i,j\in N$, 
\begin{equation*}
    v_i(A_i)+p_i \ge v_i(A_j)+p_j.
\end{equation*}
\end{definition}

Given an allocation $\mathbf{A}$, a subsidy vector $\mathbf{P}$ is \textit{envy-eliminating} if the outcome $(\mathbf{A},\mathbf{P})$ is envy-free. We denote by $\mathcal{P}(\mathbf{A})$ the set of all such subsidy vectors.
We say that $(\mathbf{A}, \mathbf{P})$ is a \textit{minimum-subsidy envy-free outcome} if it is an envy-free outcome and for every allocation $\mathbf{A'}$  (possibly $\mathbf{A'}=\mathbf{A}$) and any corresponding envy-eliminating subsidy vector $\mathbf{P'} $ in $\mathcal{P}(\mathbf{A'})$, it holds that $\sum_{i \in N}p_i  \leq \sum_{i \in N}p'_{i}$.

\begin{definition}[EFable]
\label{definition:EFable}
An allocation  $\mathbf{A}$ is called envy-freeable (EFable) if there exists a subsidy vector $\mathbf{P}$ such that the outcome $(\mathbf{A}, \mathbf{P})$ is envy-free, that is, $\mathcal{P}(\mathbf{A}) \neq \emptyset$.
\end{definition}

\begin{definition}[PO]
\label{definition:PO}
An allocation $\mathbf{A}$ is called Pareto optimal (PO) if there is no allocation $\mathbf{A'}$ such that $v_i(A_i') \ge v_i(A_i)$ for every agent  $i \in N$ and $v_j(A_j') > v_j(A_j)$ for some $j \in N$.
\end{definition}

We next introduce the envy graph to characterize EFable allocations. 

\begin{definition}[Envy Graph]
Given an allocation $\mathbf{A}=(A_1,A_2,\ldots,A_n)$, we define a complete directed graph $G_{\mathbf{A}}$ on vertex set $N$, where the weight of the directed edge $(i,j)$ is
\begin{equation*}
    w(i,j) = v_i(A_j) - v_i(A_i). 
\end{equation*}
\end{definition}
The weight $w(i,j)$ represents the amount of envy that agent $i$ has toward agent $j$. 
A path $\mathcal{F}$ in $G_{\mathbf{A}}$ is defined as a sequence of vertices $(i_0,i_1,i_2,\ldots,i_k)$. If $i_0=i_k$, then $\mathcal{F}$ is called a \emph{cycle}. 
By definition, we have $w(i,i)=0$ for every $i\in N$.
Let $w(\mathcal{F})$ denote the weight of a path $\mathcal{F}$, defined as the sum of the weights of its edges:
\begin{equation*}
     w(\mathcal{F})=\sum_{t=0}^{k-1}w(i_t,i_{t+1}). 
\end{equation*}
For any pair of agents $i,j\in N$, let $\ell(i,j)$ denote the maximum weight among all paths from $i$ to $j$. 
We further define
$\ell(i)=\max_{j\in N}\ell(i,j),$
which is the maximum weight of any path starting at $i$.

Building on this characterization, Halpern and Shah~\cite{SAGT2019} established the following theorem, which is fundamental in the study of fair allocation with subsidies.

\begin{theorem}[\textnormal{Halpern and Shah}~\cite{SAGT2019}]
\label{theorem:SAGT2019}
For an allocation $\mathbf{A}$, the following statements are equivalent.
\begin{enumerate}[label=(\arabic*).]
    \item $\mathbf{A}$ is EFable.
    \item $\mathbf{A}$ maximizes the utilitarian welfare across all reassignments of its bundles to
agents, that is, for every permutation $\sigma$ of $[n]$, $\sum_{i \in N}v_i(A_i) \ge \sum_{i \in N}v_i(A_{\sigma(i)})$. We refer to this condition as the permutation maximality.
    \item $G_{\mathbf{A}}$ has no positive-weight cycles.
\end{enumerate}
\end{theorem}

\begin{theorem}[\textnormal{Halpern and Shah}~\cite{SAGT2019}]
\label{theorem:min-subsidy-ell(i)}

For an EFable allocation $\mathbf{A}$, let $\ell(i)$ be the maximum weight of any path starting at $i$ in $G_\mathbf{A}$.
For each $i\in N$, let $p_i^*=\ell(i)$ and $\mathbf{P}^*=(p_i^*)_{i\in N}$.
Then $\mathbf{P}^*\in\mathcal{P}(\mathbf{A})$.
Moreover, for every $\mathbf{P}\in\mathcal{P}(\mathbf{A})$ and every $i\in N$,
$p_i\geq p_i^*$.

\end{theorem}

Consider the complete bipartite graph with bipartition $(N,H)$ in which edge $(i,h)$ has weight $v_i(h)$.
Since a house allocation can be regarded as an $N$-saturating matching, any permutation of the houses yields a feasible allocation.
Moreover, a maximum-weight matching maximizes total utilitarian welfare over all feasible allocations.
Therefore, the allocation induced by a maximum-weight matching maximizes utilitarian welfare among all permutations of the bundles, and hence satisfies permutation maximality.

\begin{corollary}
\label{corollary:EFable exist in bi}
An allocation is envy-freeable in house allocation problems if it corresponds to a maximum-weight matching between $N$ and $H$. 
\end{corollary}

Note that the ``only if'' direction does not necessarily hold, as house allocation does not impose non-wastefulness and thus not all houses need to be allocated.
Consider an instance with $n$ agents and $n + 1$ houses, where every agent has an identical utility function: each agent values one particular house at $1$ and all other $n$ houses at $0$. 
If the $n$ houses of value $0$ are allocated to the $n$ agents, the resulting allocation is still EFable and the corresponding subsidies are $\mathbf{0}$.

\section{Binary Utility Functions}
\label{section:Binary Utility Functions}

In this section, we consider the binary instances, in which we have $v_i(h) \in \{0,1\}$ for all $i\in N$ and $h\in H$.
Given a binary instance $\mathcal{I}$, it is often more convenient to work with the bipartite graph $G_\mathcal{I}=(N \cup H,E)$, where we put an edge $(i,h)\in E$ between agent $i$ and house $h$ if and only if $v_i(h) = 1$.
We call the  bipartite graph $G_\mathcal{I}$ a representing graph of the instance $\mathcal{I}$.
Every allocation for the instance can be represented by a matching in the representing graph, where the matched agents receive houses of value $1$ and each unmatched agent receives an unmatched house (which has value $0$ to the agent).
Throughout, we call an agent \emph{matched} under allocation $\mathbf{A}$ if $v_i(A_i)=1$.

\subsection{Structural Properties}

\begin{observation}
\label{observation:is-1}
    For binary instances, the weight $w(i,j)$ in the envy graph is in $\{-1,0,1\}$ for any $i, j \in N$.  
Specifically, for any matched agent $i$, $w(i,j) \in \{0,-1\}$ for any $j \in N$;  
for any unmatched agent $i$, $w(i,j) \in \{0,1\}$ for any $j \in N$.
\end{observation}

\begin{lemma}
\label{subsidy-is-1}
    For binary instances, each agent requires a subsidy of at most $1$ in any EFable allocation $\mathbf{A}$, i.e., $\ell(i) \leq 1$ for any $i \in N$.

\end{lemma}
\begin{proof}
    We prove the lemma by contradiction. 
    Assume that there exists an agent $i$ who requires $p_i > 1$, i.e., $\ell(i) > 1$.  
    Suppose $j=\arg\max_{j \in N} \ell(i,j)$. Since there is also an edge from $j$ to $i$, and this edge has weight at least $-1$ by Observation~\ref{observation:is-1}, the path defining $\ell(i,j)$, together with edge $(j,i)$ form a cycle.
    The total weight of this cycle is $\ell(i) + w(j,i) > 1 + (-1) = 0$, which contradicts the characterization given in Theorem~\ref{theorem:SAGT2019}.
\end{proof}

\begin{corollary}
    A total subsidy of at most $(n-1)$ suffices to guarantee envy-freeness, and this bound is tight.

\end{corollary}
\begin{proof}
In any minimum total subsidy vector $\mathbf{P}$, at least one agent must receive zero subsidy; otherwise, all subsidies could be uniformly reduced while preserving envy-freeness, contradicting the minimality of $\mathbf{P}$.
Lemma~\ref{subsidy-is-1} implies that, to achieve envy-freeness under binary valuations in house allocation, the total required subsidy is at most $(n-1)$.
To show that this bound is tight, we consider an instance with identical binary valuations where each agent values a particular house at $1$ and all other houses at $0$, and assume $m = n$. 
In any allocation, this special house is assigned to exactly one agent. 
To achieve envy-freeness, every other agent must receive a subsidy of $1$. 
Thus, a total subsidy of $n - 1$ is required. 
\end{proof}

    The following lemma states that when the minimum subsidy is non-zero, the corresponding allocation forms a maximum cardinality matching in $G_\mathcal{I}$ for the binary instance $\mathcal{I}$.
    Note that this lemma holds only when the amount of minimum total subsidy is nonzero; in the example at the end of Section~\ref{section:Preliminaries}, the allocation requires $\mathbf{0}$ subsidy but admits an augmenting path.
    
\begin{lemma}
\label{subsidy-is-0}
   If $(\mathbf{A^*}, \mathbf{P^*})$ is a minimum-subsidy envy-free outcome and $\sum_{i \in N} p^*_i > 0$, then the matching corresponding to allocation $\mathbf{A^*}$ in $G_\mathcal{I}$ does not admit any augmenting path.
   Equivalently, it is a maximum-cardinality matching and $\mathbf{A^*}$ is Pareto optimal.
\end{lemma}
\begin{proof}
    We prove the lemma by contradiction.
    Suppose that $(\mathbf{A^*}, \mathbf{P^*})$ is a minimum-subsidy envy-free outcome with $\sum_{i \in N} p^*_i > 0$, and allocation $\mathbf{A^*}$ admits an augmenting path in $G_\mathcal{I}$.
    Since the total subsidy $\sum_{i \in N}p_{i}^{*}>0$, there exists an agent $k$ who receives a subsidy of $1$ by Lemma~\ref{subsidy-is-1}, i.e., $p^*_k =1$.
    It means that each unmatched agent should also get a subsidy of $1$ to maintain envy-freeness.
    Assume that the allocation $\mathbf{A^*}$ admits an augmenting path from agent $j$ to house $h^*$, as shown in Figure~\ref{fig:augmenting-path}. 
    It follows that agent $j$ is unmatched and house $h^*$ is unassigned in $\mathbf{A^*}$, which implies $p^*_j = 1$.
     By applying the augmenting path from agent $j$ to house $h^*$, we can construct a new allocation $\mathbf{A}$ in which all agents matched in $\mathbf{A^*}$ remain matched, and agent $j$ becomes matched in allocation $\mathbf{A}$.
     
    Next, we prove that $(\mathbf{A}, \mathbf{P})$ is an EF outcome by defining subsidy vector $\mathbf{P}$ as $p_j = 0$ and $p_i = p^*_i$ for all $i \in N \setminus \{j\}$.
    Let $N_0$ denote the unmatched agents and $H_0$ denote the houses allocated to $N_0$ in allocation $\mathbf{A^*}$.
    Since allocation $\mathbf{A^*}$ is EFable, we have $\forall h \in H_0,\ v_j(h)=0$ and $\mathbf{A^*}$ satisfies permutation maximality.
    
    Next we prove that there does not exist augmenting path from $N_0$ to $H_0$ with respect to allocation $\mathbf{A}$.
    Assume, for the sake of contradiction, that there exists an augmenting path from $N_0$ to $H_0$.
    Since the set of houses whose assignments differ between $\mathbf{A}$ and $\mathbf{A}^*$ contains exactly one house, this augmenting path must contain the newly matched house $h^*$ on which $\mathbf{A}$ and $\mathbf{A}^*$ differ.
    Otherwise, the same augmenting path would also exist in $\mathbf{A}^*$, implying that $\mathbf{A}^*$ does not satisfy permutation maximality.
    This contradicts the fact that $\mathbf{A}^*$ is EFable.
    Let the augmenting path in $\mathbf{A}$ be
\[
(i_0, h_1, i_1, h_2, i_2, \ldots, h^*,j' \ldots, h_0),
\]
where $i_0 \in N_0$ and $h_0 \in H_0$.
Since the node $h^*$ is the only house on which $\mathbf{A}$ and $\mathbf{A}^*$ differ, adding the alternating path $(j',\ldots,j)$, as shown in Figure~\ref{fig:augmenting-path}, yields a new path $(j,\ldots,j',\ldots,h_0)$, which is an alternating path with respect to $\mathbf{A}^*$.
Moreover, this new path starts at an unmatched agent $j$ and ends at a house $h_0 \in H_0$ that is unmatched in $\mathbf{A}^*$.
Hence, $(j, \ldots, h_0)$ constitutes an augmenting path in $\mathbf{A}^*$, which contradicts the fact that $\mathbf{A}^*$ satisfies permutation maximality.
Therefore, there is no augmenting path from $N_0$ to $H_0$ in $\mathbf{A}$. 
This implies that the allocation $\mathbf{A}$ satisfies permutation maximality and is EFable by Theorem~\ref{theorem:SAGT2019}.


    For the agents on the augmenting path, if in allocation $\mathbf{A}^*$ some matched agent on the path receives a subsidy of $1$, then all matched agents on the path must receive a subsidy of $1$ to maintain EF in $\mathbf{A}^*$. 
    Therefore, in $(\mathbf{A}, \mathbf{P})$, these matched agents are in exactly the same situation as in $\mathbf{A}^*$, and their final utilities remain unchanged from the perspective of other agents.
    Indeed, in $\mathbf{A}$, only the situation of agent $j$ has changed, so we only need to consider agent $j$. 
    For all $ i \in N \setminus \{j\},$ 
    \begin{equation*}
     \ v_j(A_j)+p_j=1=v_j(A^*_j)+p^*_j \geq v_j(A^*_i)+p^*_i=v_j(A_i)+p_i,   
    \end{equation*}
   implying agent $j$ does not envy any other agent. 
   For all $ i \in N \setminus \{j\},$
    \begin{align*}
      v_i(A_i)+p_i &= v_i(A^*_i)+p^*_i\geq v_i(A^*_j)+p^*_j=1 \\
      &\geq v_i(A_j)=v_i(A_j)+p_j,
    \end{align*} 
   implying no agent envies agent $j$.  
   
    Thus, $(\mathbf{A}, \mathbf{P})$ is an EF outcome with less subsidy than $(\mathbf{A^*}, \mathbf{P^*})$, leading to a contradiction.
Hence, the matching induced by $\mathbf{A^*}$ has no augmenting path and is maximum-cardinality. Under binary utilities, any Pareto improvement would strictly increase the number of agents receiving a house of value $1$, contradicting maximum cardinality. Therefore, $\mathbf{A^*}$ is Pareto optimal.
\end{proof}

\begin{figure}[t]
    \captionsetup{font=footnotesize}
    \centering
    \begin{tikzpicture}
     [
        node distance=1.5cm,
        on grid,
        thick,
        font=\small]
     
    \node
    [ 
        state,
        fill=orange!20,
        align=center,
        inner sep=1mm,
        minimum size=3mm,
    ] (q_0) [label=above: agent $j'$]{};
     
    \node
    [
        state,
        fill=orange!20,
        align=center,
        inner sep=1mm,
        minimum size=3mm,
    ] (q_1) [right=of q_0, label=left: ] {};
     
     \node
    [ 
        state,
        fill=orange!20,
        align=center,
        inner sep=1mm,
        minimum size=3mm,
    ] (q_2) [right=of q_1, label=left:]{};

    \node 
    [ 
        state,
        fill=green!20,
        align=center,
        inner sep=1mm,
        minimum size=3mm,
   ] (q_7) [below= of q_0, label=below: house $h^*$ ]{};
     
     \node 
    [ 
        state,
        fill=green!20,
        align=center,
        inner sep=1mm,
        minimum size=3mm,
   ] (q_3) [right= of q_7]{};

      \node
    [ 
        state,
        fill=green!20,
        align=center,
        inner sep=1mm,
        minimum size=3mm,
    ] (q_4) [right=of q_3]{};
    
    \node
    [ 
        state,
        fill=green!20,
        align=center,
        inner sep=1mm,
        minimum size=3mm,
    ] (q_5) [right=of q_4]{};

  \node
    [ 
        state,
        fill=green!20,
        align=center,
        inner sep=1mm,
        minimum size=3mm,
    ] (q_8) [right=of q_5, label=below: house $A^{*}_j$ ]{};
    
    \node
    [ 
        state,
        fill=orange!20,
        align=center,
        inner sep=1mm,
        minimum size=3mm,
    ] (q_6) [right=of q_2, label=above: agent $j$]{};
        
    \path [-]
        (q_0) edge [bend left=0, color=blue]
            node [above] {} (q_7)
        (q_3) edge [bend left=0, color=blue]
            node [above] {} (q_1)
        (q_4) edge [bend left=0, color=blue]   
            node [above] {} (q_2)
        (q_5) edge [bend left=0, color=blue]   
            node [above] {}(q_6)
            ;
        \path [-]
        (q_3) edge [bend left=0] 
            node [above] {} (q_0)
        (q_4) edge [bend left=0]   
            node [above] {} (q_1)
        (q_5) edge [bend left=0]   
            node [above] {} (q_2)
            ;
        \path [-,dashed]
        (q_6) edge [bend left=0] 
            node [above] {} (q_8)
            ;
    \end{tikzpicture}
    \captionsetup{font=small}
    \caption{Example of an augmenting path from agent $j$ to house $h^*$, where orange nodes denote agents, green nodes denote houses, black edges denote assignments in $\mathbf{A^*}$ (solid if the agent values the house at $1$, dashed if at $0$), and blue edges denote assignments in $\mathbf{A}$.
}
    \label{fig:augmenting-path}
    \end{figure}
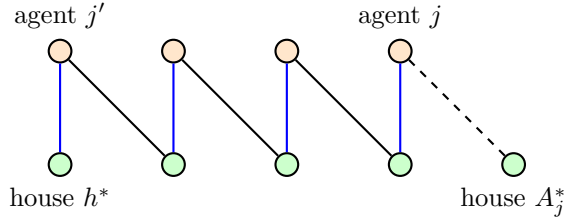

Lemma~\ref{subsidy-is-0} implies that the allocation in a minimum subsidy envy-free outcome should correspond to a maximum cardinality matching with respect to $N$ and $H$. 
However, since the number of maximum cardinality matchings may not be bounded by $\text{poly}(n,m)$, we cannot guarantee finding the one that minimizes the total subsidy in polynomial time.
Since allocating unmatched houses to unmatched agents requires giving each of them a subsidy of $1$, such allocations may trigger the need for subsidies among matched agents if they value these unmatched houses at $1$, and this effect can further propagate within the set of matched agents. 
Therefore, determining which unmatched houses should be assigned to the unmatched agents can itself be a hard problem.
Indeed, we will next show that \problem{Minimum Subsidy for Binary} is NP-complete.

\subsection{Hardness}

Next, we show that minimizing the total subsidy required to achieve envy-freeness is NP-hard even under binary utilities, unless $\mathrm{P}=\mathrm{NP}$.


  \begin{figure*}[t]
    \captionsetup{font=footnotesize}
    \centering
    \begin{tikzpicture}
     [
        node distance=1.4cm,
        on grid,
        thick,
        font=\small]
     
    \node
    [ 
        state,
        fill=orange!20,
        align=center,
        inner sep=1mm,
        minimum size=3mm,
    ] (q_0) [label=above: $a^*_1$]{};
     
    \node
    [
        state,
         fill=green!20,
        align=center,
        inner sep=1mm,
        minimum size=3mm,
    ] (q_1) [below=of q_0, label=below: $h_{a^*_1}$] {};

     \node
    [ 
        state,
         fill=orange!20,
        align=center,
        inner sep=1mm,
        minimum size=3mm,
    ] (q_3) [right=of q_0, label=above: $a^*_2$]{};
     
      \node
    [ 
        state,
         fill=green!20,
        align=center,
        inner sep=1mm,
        minimum size=3mm,
    ] (q_4) [below=of q_3, label=below: $h_{a^*_2}$]{};
     
      \node
    [
        state,
         fill=orange!20,
        align=center,
        inner sep=1mm,
        minimum size=3mm,
    ] (q_2) [right=of q_3, label=above: $a^*_3$] {};

    \node
    [
        state,
        fill=orange!20,
        align=center,
        inner sep=1mm,
        minimum size=3mm,
    ] (q_5) [right=of q_2, label=above: $a^*_4$] {};

    \node
    [
        state,
         fill=orange!20,
        align=center,
        inner sep=1mm,
        minimum size=3mm,
    ] (q_6) [right=of q_5, label=above: $a^*_5$] {};

  \node
    [ 
        state,
         fill=green!20,
        align=center,
        inner sep=1mm,
        minimum size=3mm,
    ] (q_7) [below=of q_2, label=below: $h_{a^*_3}$]{};

    \node
    [ 
        state,
         fill=green!20,
        align=center,
        inner sep=1mm,
        minimum size=3mm,
    ] (q_8) [below=of q_5, label=below: $h_{a^*_4}$]{};

    \node
    [ 
        state,
         fill=green!20,
        align=center,
        inner sep=1mm,
        minimum size=3mm,
    ] (q_9) [below=of q_6, label=below: $h_{a^*_5}$]{};

     \node
    [
        state,
        fill=orange!20,
        align=center,
        inner sep=1mm,
        minimum size=3mm,
    ] (q_10) [right=of q_6, label=above: $a^*_6$] {};

    \node
    [ 
        state,
         fill=green!20,
        align=center,
        inner sep=1mm,
        minimum size=3mm,
    ] (q_11) [below=of q_10, label=below: $h_{a^*_6}$]{};

    \node
    [
        state,
        fill=orange!20,
        align=center,
        inner sep=1mm,
        minimum size=3mm,
    ] (q_12) [right=2.4cm of q_10, label=above: $a_1$] {};

     \node
    [
        state,
        fill=orange!20,
        align=center,
        inner sep=1mm,
        minimum size=3mm,
    ] (q_13) [right= of q_12, label=above: $a_2$] {};

    \node
    [
        state,
        fill=orange!20,
        align=center,
        inner sep=1mm,
        minimum size=3mm,
    ] (q_14) [right= of q_13, label=above: $a_3$] {};

    \node
    [
        state,
        fill=orange!20,
        align=center,
        inner sep=1mm,
        minimum size=3mm,
    ] (q_15) [right= of q_14, label=above: $a_4$] {};

     \node
    [
        state,
        fill=orange!20,
        align=center,
        inner sep=1mm,
        minimum size=3mm,
    ] (q_16) [right= of q_15, label=above: $a_5$] {};

     \node
    [
        state,
        fill=green!20,
        align=center,
        inner sep=1mm,
        minimum size=3mm,
    ] (q_17) [right= of q_11, label=below: $h^*_{1}$] {};

     \node
    [
        state,
        fill=green!20,
        align=center,
        inner sep=1mm,
        minimum size=3mm,
    ] (q_18) [right= of q_17, label=below: $h^*_{2}$] {};

     \node
    [
        state,
        fill=green!20,
        align=center,
        inner sep=1mm,
        minimum size=3mm,
    ] (q_19) [right= of q_18, label=below: $h^*_{3}$] {};

      \node
    [
        state,
        fill=green!20,
        align=center,
        inner sep=1mm,
        minimum size=3mm,
    ] (q_20) [right= of q_19, label=below: $h_1$] {};

     \node
    [
        state,
        fill=green!20,
        align=center,
        inner sep=1mm,
        minimum size=3mm,
    ] (q_21) [right= of q_20, label=below: $h_2$] {};

    \node
    [
        state,
        fill=green!20,
        align=center,
        inner sep=1mm,
        minimum size=3mm,
    ] (q_22) [right= of q_21, label=below: $h_3$] {};
    
    \path [-]
         (q_0) edge [bend left=0] 
            node [above] {} (q_1)  
         (q_5) edge [bend left=0]   
            node [above] {} (q_8) 
        (q_3) edge [bend left=0]   
            node [above] {} (q_4)
        (q_2) edge [bend left=0]   
            node [above] {} (q_7)
        (q_6) edge [bend left=0]   
            node [above] {} (q_9)
        (q_10) edge [bend left=0]   
            node [above] {} (q_11) 
        (q_0) edge [bend left=0]   
            node [above] {} (q_17)
        (q_3) edge [bend left=0]   
            node [above] {} (q_17)
        (q_2) edge [bend left=0]   
            node [above] {} (q_18)
        (q_5) edge [bend left=0]   
            node [above] {} (q_18)
        (q_6) edge [bend left=0]   
            node [above] {} (q_18)
        (q_5) edge [bend left=0]   
            node [above] {} (q_19)
        (q_6) edge [bend left=0]   
            node [above] {} (q_19)
        (q_10) edge [bend left=0]   
            node [above] {} (q_19)

         (q_12) edge [bend left=0]   
            node [above] {} (q_20)
        (q_13) edge [bend left=0]   
            node [above] {} (q_20)
        (q_14) edge [bend left=0]   
            node [above] {} (q_20)
         (q_15) edge [bend left=0]   
            node [above] {} (q_20)
        (q_16) edge [bend left=0]   
            node [above] {} (q_20)

        (q_12) edge [bend left=0]   
            node [above] {} (q_21)
        (q_13) edge [bend left=0]   
            node [above] {} (q_21)
        (q_14) edge [bend left=0]   
            node [above] {} (q_21)
         (q_15) edge [bend left=0]   
            node [above] {} (q_21)
        (q_16) edge [bend left=0]   
            node [above] {} (q_21)

        (q_12) edge [bend left=0]   
            node [above] {} (q_22)
        (q_13) edge [bend left=0]   
            node [above] {} (q_22)
        (q_14) edge [bend left=0]   
            node [above] {} (q_22)
         (q_15) edge [bend left=0]   
            node [above] {} (q_22)
        (q_16) edge [bend left=0]   
            node [above] {} (q_22)
            
      
            ;



  \node[fit=(q_0) (q_10),draw=none,inner sep=0pt] (group) {};

  \draw [decorate,decoration={brace,amplitude=6pt}] 
        ([yshift=15pt]group.north west) -- ([yshift=15pt]group.north east)
        node[midway,yshift=12pt] {$N_1$};

    \node[fit=(q_12) (q_16),draw=none,inner sep=0pt] (group) {};

  \draw [decorate,decoration={brace,amplitude=6pt}] 
        ([yshift=15pt]group.north west) -- ([yshift=15pt]group.north east)
        node[midway,yshift=12pt] {$N_2$};

  \node[fit=(q_1) (q_11),draw=none,inner sep=0pt] (group) {};

  \draw [decorate,decoration={brace,amplitude=6pt,mirror}] 
        ([yshift=-25pt]group.north west) -- ([yshift=-25pt]group.north east)
        node[midway,yshift=-12pt] {$H_1$};

  \node[fit=(q_17) (q_19),draw=none,inner sep=0pt] (group) {};

  \draw [decorate,decoration={brace,amplitude=6pt,mirror}] 
        ([yshift=-25pt]group.north west) -- ([yshift=-25pt]group.north east)
        node[midway,yshift=-12pt] {$H_2$};

  \node[fit=(q_20) (q_22),draw=none,inner sep=0pt] (group) {};

  \draw [decorate,decoration={brace,amplitude=6pt,mirror}] 
        ([yshift=-25pt]group.north west) -- ([yshift=-25pt]group.north east)
        node[midway,yshift=-12pt] {$H_3$};

    \end{tikzpicture}
    \captionsetup{font=small}
    \caption{An example of the reduction construction from the \problem{Minimum $k$-Union} instance ($U = \{1,2,3,4,5,6\}$, $S_1 = \{1,2\}$, $S_2 = \{3,4,5\}$, $S_3 = \{4,5,6\}$, $k=2$, $q=4$) to the corresponding \problem{Minimum Subsidy for Binary} instance.}
    \label{figure:reduction}
    \end{figure*}
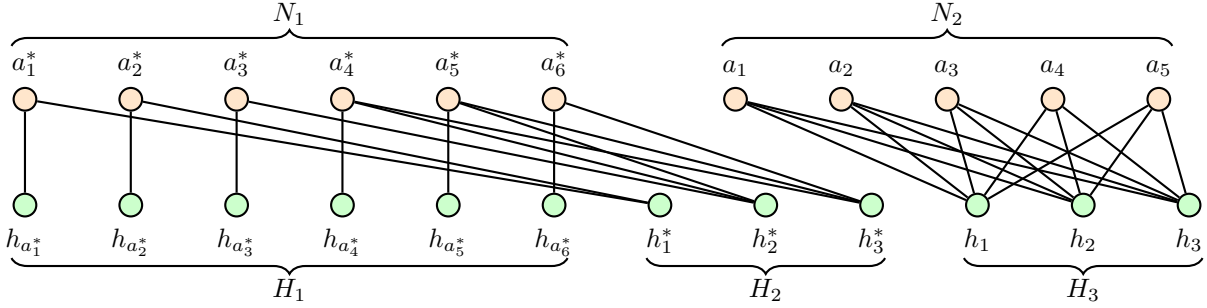

\begin{theorem}
[Hardness of \problem{Minimum Subsidy for Binary}]
\label{theorem:reduction}
   The problem of determining whether for a given positive integer $\gamma>0$, there exists an envy-free outcome $(\mathbf{A,P})$ such that $\sum_{i \in N}p_i \leq \gamma$, is NP-complete even when all agents have binary utilities.
\end{theorem}

\begin{proof}[Proof Sketch.]
    Since for each envy-free outcome $(\mathbf{A,P})$, we can compute $\sum_{i \in N}p_i$ in polynomial time, this problem is in NP.
    We prove the NP-hardness of this problem by reducing from the decision version of \problem{Minimum $k$-Union} problem. 
    In this problem, we are given a finite set $U:=\{e_1,e_2, \cdots, e_n\}$ of elements, subsets $S_1, S_2, \ldots, S_m$ of $U$, and positive integers $q, k$ such that $q \le |U|$ and $k < m$; 
the goal is to determine whether there exists a subset $T \subseteq [m]$ such that $|T| = k$ and $\left| \bigcup_{t \in T} S_t \right| \le q$.
It is known that this problem is NP-complete~\cite{vinterbo2002note}.

Suppose that we are given an instance of the decision version of \problem{Minimum $k$-Union}. Then we construct an instance of our house allocation problem as follows.
\begin{itemize}
    \item Define the agent set $N := N_1 \cup N_2$, where $N_1:=\{a^*_e \mid e\in U\}, N_2:=\{a_t \mid t\in [m+k]\}$;
    \item Define the house set $H := H_1 \cup H_2 \cup H_3$, where $H_1:=\{h_{a^*_e} \mid e\in U\}, H_2:=\{h^*_j \mid j \in [m]\} \ \text{and}\ H_3:=\{h_t \mid t\in [m]\}$.
    \item For each element $e \in U$, define the utility function $v_{a^*_e}: H \to \{0,1\}$ by
    \[
    v_{a^*_e}(h) :=
    \begin{cases}
        1 & \text{if }  h=h_{a^*_e}; \\
        1 & \text{if } h = h^*_j \text{ such that } e \in S_j;\\
        0 & \text{otherwise.}
    \end{cases}
    \]
    \item For each integer $t \in [m+k]$, define the utility function $v_{a_{t}}: H \to \{0,1\}$ by
    \[
    v_{a_t}(h) :=
    \begin{cases}
        1 & \text{if }  h \in H_3; \\
        0 & \text{otherwise.}
    \end{cases}
    \]
\end{itemize}

Clearly, the construction takes polynomial time. 
Note that there are $(n + m +k)$ agents and $(n + m + m)$ houses, and the number of houses exceeds the number of agents.

For example, if $U = \{1,2,3,4,5,6\}$, $S_1 = \{1,2\}$, $S_2 = \{3,4,5\}$, $S_3 = \{4,5,6\}$, $k=2$, $q=4$, then the corresponding reduced house allocation instance is given as in Figure~\ref{figure:reduction}.

As shown in Figure~\ref{figure:reduction}, at least one house in $H_3$ must be allocated, and at least $k$ agents in $N_2$ cannot obtain houses from $H_3$, which causes each of these $k$ agents to require a subsidy of $1$.
Allocating houses in $H_1$ to the corresponding agents in $N_1$ is optimal, since such allocations do not cause envy among other agents.
Hence, we need to assign $k$ houses from $H_2$ to these $k$ agents in $N_2$.
However, houses in $H_2$ have value $1$ for some agents in $N_1$, implying that these agents must also receive a subsidy of $1$ to ensure envy-freeness with respect to the $k$ agents in $N_2$.
Thus, finding an allocation with a total subsidy of $k+q$ is equivalent to selecting $k$ houses from $H_2$ whose corresponding union set in $N_1$ has size $q$ (see Appendix~\ref{appendix:Missing-Proof-of-NP-hard} for the detailed proof).
\end{proof}

We remark that \problem{Minimum $k$-Union} has known hardness of $O(n^{1/4})$-approximation under ``Dense versus Random'' conjecture for DkS to hypergraphs~\cite{chlamtavc2017minimizing}.
Indeed, the reduction in Theorem~\ref{theorem:reduction} can be modified to be approximation-preserving, which directly yields the hardness of approximation for \problem{Minimum Subsidy for Binary} (see Appendix~\ref{Appendix:approximation-preserving} for the full proof).

\section{\texorpdfstring{$k$}{k} Types of Agents}
\label{section:k-types}


In this section, we consider the case with $k$ types of agents having binary utilities. 
Suppose that the agent set is partitioned into $N_1 \cup N_2 \cup \cdots \cup N_k$, where agents in $N_i$ have the same binary valuation function $v_i: H \to \{0,1\}$.
Lemma~\ref{subsidy-is-0} implies that the potential optimal EFable allocation should be maximum cardinality matching.
It remains to identify a maximum cardinality matching for which the minimum number of agents require subsidies.

Next, we introduce the well-known Dulmage–Mendelsohn decomposition~\cite{Dulmage_Mendelsohn_1958}, which partitions the vertices of a bipartite graph into subsets with the property that two adjacent vertices belong to the same subset if and only if they are matched together in some perfect matching of the graph.
Aigner-Horev and Segal-Halevi~\cite{uniqueAigner-HorevS22} extended this concept to the so-called EFM partition, which can be used to identify envy-free matchings.
We reformulate their main theorem accordingly.

\begin{theorem}[\textnormal{Aigner-Horev and Segal-Halevi}~\cite{uniqueAigner-HorevS22}]
\label{theorem-EFM-partition}
    Every bipartite graph $G=(X\cup Y, E)$ admits a unique partition $X=X_S \cup X_L$ and $Y=Y_S \cup Y_L$ satisfying the following three conditions: 
    \begin{enumerate}[label=(\alph*)]
        \item There are no edges between $X_S$ and $Y_L$;
        \item The subgraph $G[X_S,Y_S]$ induced by $X_S$ and $Y_S$ admits a $Y_S$-saturating matching;
        \item The subgraph $G[X_L, Y_L]$ induced by $X_L$ and $Y_L$ admits an $X_L$-saturating matching.
    \end{enumerate}
\end{theorem}

For any binary instance $\mathcal{I}$, the corresponding bipartite graph $G_\mathcal{I}=(N\cup H, E)$ can be uniquely partitioned into sets satisfying the conditions of Theorem~\ref{theorem-EFM-partition}, where the agent set $N$ corresponds to $X$ and the house set $H$ corresponds to $Y$.
The unique EFM partition can be found in polynomial time~\cite{uniqueAigner-HorevS22}.
Based on the unique EFM partition $N=N_S \cup N_L$ and $H=H_S \cup H_L$ of $G_\mathcal{I}$, we then characterize the properties that the EFable allocation induced by a maximum cardinality matching must satisfy in order to achieve the minimum total subsidy (with full proof provided in Appendix~\ref{appendix-section:k-types}).

\begin{lemma}
\label{lemma:min-subsidy-conditions}
The EFable allocation $\mathbf{A^*}$ induced by a maximum cardinality matching $\mathcal{M}$ in $G_\mathcal{I}$ that requires a minimum amount of subsidy must have the following structure:
\begin{enumerate}[label=(\arabic*).]
    \item The agent set $N_{S}$ whose neighbor set $H_{S}$ in $G_\mathcal{I}$ satisfies $|H_{S}|<|N_{S}|$, and all houses in $H_{S}$ are allocated to agents in $N_{S}$. 
    Let $t=|N_{S}|-|H_{S}|$ be the number of unmatched agents, who are allocated houses they do not like, and receive a subsidy of 1.
    \item The remaining agents in $N_{L}=N\setminus N_{S}$ are all matched.
    \item There is a subset $\hat{N}_{L}\subseteq N_{L}$ such that each agent in $\hat{N}_{L}$ receives a subsidy of $1$.
    \item Agents of the same type belong to exactly one of $N_{S}, \hat{N}_{L}, N_{L}\setminus\hat{N}_{L}$. 
\end{enumerate}
\end{lemma}

\begin{observation}
\label{observation:partitions-number}
There are at most $2^{k}$ possible partitions $(N_S,\hat{N}_L, N_L \setminus \hat{N}_L)$ of the agents.
\end{observation}
\begin{proof}
By condition~(4) of Lemma~\ref{lemma:min-subsidy-conditions}, agents of the same type belong to exactly one of $N_S$, $\hat{N}_L$, and $N_L \setminus \hat{N}_{L}$.
Since $N_S$ can be efficiently identified via the unique EFM partition, each of the remaining at most $k$ agent types must be assigned either to $\hat{N}_L$ or to $N_L \setminus \hat{N}_{L}$.
Hence, the total number of possible partitions is bounded by $2^{k}$.
\end{proof}


In particular, the partition $(N_{S}, \hat{N}_{L}, N_{L}\setminus\hat{N}_{L})$ corresponding to the minimum total subsidy allocation $\mathbf{A}^*$ must be among the $2^{k}$ possible partitions.
For each such partition, we then try to construct the EFable allocations with subsidies that meet the conditions of Lemma~\ref{lemma:min-subsidy-conditions}.

\paragraph{General Idea.} 
Given a partition $(N_{S}, \hat{N}_{L}, N_{L}\setminus\hat{N}_{L})$, the total subsidy induced by this partition is $|N_S|-|H_S|+|\hat{N}_{L}|$, according to properties (1) and (3) of Lemma~\ref{lemma:min-subsidy-conditions}.
Therefore, determining whether this total subsidy is achievable reduces to verifying the feasibility of the partition.
Specifically, we can examine whether every agent can be assigned a house under the construction described in Lemma~\ref{lemma:min-subsidy-conditions} to verify the feasibility.
If such a partition is feasible, then an EFable allocation $\mathbf{A}$ with subsidy $|N_S|-|H_S|+|\hat{N}_{L}|$ can be constructed according to a maximum cardinality matching which satisfies conditions of Lemma~\ref{lemma:min-subsidy-conditions}. 
By checking the feasibility of all $2^{k}$ possible partitions, we can determine the minimum achievable total subsidy and construct a corresponding optimal allocation $\mathbf{A}^*$.
To this end, we update the bipartite graph to try to construct an allocation satisfying the conditions of Lemma~\ref{lemma:min-subsidy-conditions}.

\begin{observation}
\label{observation:types-of-agents}
There are at most $2^{k}$ types of houses.
\end{observation}
\begin{proof}
The type of a house can be represented by a $k$-dimensional binary vector $\mathbf{c} \in \{0,1\}^k$, where the $i$-th entry equals $1$ if agents in $N_i$ have a value of $1$ to the house, and 
$0$ otherwise.
Thus, the total number of house types is $2^k$.
\end{proof}

Suppose that the house set is partitioned into $H_1 \cup H_2 \cup \cdots \cup H_{2^k}$, where houses in $H_j$ have the same house type.
Note that the agent set is partitioned into $N_1 \cup N_2 \cup \cdots \cup N_k$, where agents in $N_i$ have the same binary valuation function.
Thus, we can construct a new bipartite graph $G^*=(U \cup V,E)$ for the instance with $k$ types of agents, where there are $k$ vertices in $U$ on the left side of the graph representing the $k$ types of agents; there are $2^{k}$ vertices in $V$ on the right side of the graph representing the $2^{k}$ types of houses.
We put an edge $ e\in E$ between a type of agent and a type of house if the agents value the houses at $1$.
Moreover, the capacity of each agent type (resp.\ house type) is defined as the number of agents (resp.\ houses) of that type.
The capacity of each agent node ensures that agents of the corresponding type are matched exactly as many times as there are agents of that type. 
Similarly, the capacity of each house node ensures that no more houses of that type are assigned than are available.

\paragraph{The Simplified Algorithm.} Our algorithm begins with a given partition $(N_{S}, \hat{N}_{L}, N_{L}\setminus\hat{N}_{L})$  and the new bipartite graph $G^*=(U \cup V,E)$. 
We perform the update operations on $G^*$. 
From conditions (1) and (2) of Lemma~\ref{lemma:min-subsidy-conditions}, any optimal allocation assigns all houses in $H_S$ to agents in $N_S$, leaving $t$ unmatched agents. 
Accordingly, we remove $H_S$ and $N_S$ from $G^*$ and construct a new agent type vertex representing the $t$ unmatched agents.
Since each agent in $\hat{N}_L$ and the $t$ unmatched agents must receive a subsidy of $1$, houses valued at $1$ by agents in $N_L\setminus \hat{N}_L$ (denoted by $H^*$) cannot be assigned to them.
Otherwise, an agent in $N_L\setminus \hat{N}_L$ would envy an agent receiving both a house of value $1$ and a subsidy of $1$, while obtaining only a house of value $1$ herself.
Therefore, we delete all edges between $\hat{N}_L$ and $H^*$, as well as between the $t$ unmatched agents and $H^*$.
Finally, in the resulting graph $G^*$, if there exists a matching that assigns a house to every agent in $N_L$ and assigns each of the $t$ unmatched agents a house outside $H^*$, then there exists an allocation satisfying the conditions of Lemma~\ref{lemma:min-subsidy-conditions}, and the partition is feasible.
A detailed description of the algorithm, along with its analysis, is provided in Appendix~\ref{appendix-section:k-types}.
Based on the detailed analysis, we obtain the following main theorem of this section.

 \begin{theorem}
 \label{theorem:k-types}
Given a binary instance $\mathcal{I} = (N, H, \mathbf{v})$, if there are $k$ types of agents, then an EFable allocation with minimum total subsidy can be found in $O((k \cdot 2^{2k})^{(1+o(1))}) +O(n + m) $ time.
\end{theorem}

\section{Two Types of Agents with General Utilities}
\label{section:two-types}
In this section, we present a polynomial time algorithm for the case with two types of agents, where agents of the same type have the same utility function.
The set of agents is partitioned into $N_x$ (agents of type-$x$) and $N_y$ (agents of type-$y$), where all agents in $N_x$ have utility function $v_x$; all agents in $N_y$ have utility function $v_y$.
Let $n_x = |N_x|$ and $n_y = |N_y|$ be the number of type-$x$ and type-$y$ agents, respectively.
We have $n_x + n_y = n \leq m$.

By condition~(3) of Theorem~\ref{theorem:SAGT2019}, an allocation is EFable if and only if the corresponding envy graph contains no positive-weight cycles. 
Note that the weight of any cycle formed entirely by agents of the same type is zero, since all agents of the same type have identical utility functions. 
Therefore, to check for positive-weight cycles, it suffices to consider only cycles involving agents of different types. 
Furthermore, in such cross-type cycles, the weight of each edge equals the difference in valuations between the two types, which is fully captured by the differences in house values. 
Therefore, we define $d(h) = v_x(h) - v_y(h)$ for each house $h \in H$, and sort the houses in descending order of $d(h)$, reindexing them accordingly such that  
$d(h_1) \;\geq\; d(h_2) \;\geq\; \cdots \;\geq\; d(h_m),$ 
and let $\mathcal{D} := \{d(h_1), d(h_2), \ldots, d(h_m)\}$.

\begin{example}
    Consider the following instance with two agents with type-$x$ and two agents with type-$y$ and $6$ houses. The agents’ utility functions are shown in Table~\ref{table:two-types}.

    \begin{table}[h]
     \centering
    \begin{tabular}{ccccccc}\toprule
		 & $h_1$ & $h_2$ & $h_3$ & $h_4$ & $h_5$ & $h_6$  \\ 
        \midrule
		$v_x$ & $6$ & $3$ & $2$  & $3$ & $1$ & $2$  \\
		$v_y$ & $2$ & $1$ &  $1$  & $2$ & $1$ & $5$ \\
           $d(h)$ & $4$ & $2$ &   $1$ & $1$ & $0$ & $-3$ \\
		 \bottomrule
	\end{tabular}
    \caption{An example of two types of agents.}
    \label{table:two-types}
\end{table}

It is easy to verify that the allocation $\mathbf{A} = (A_x = \{h_2, h_4\}, A_y = \{h_3, h_5\})$ is EFable and requires a total subsidy of $2$, obtained by giving a subsidy of $1$ to each agent in $N_y$.
Another allocation $\mathbf{A'} = (A_x = \{h_1, h_2\}, A_y = \{h_5, h_6\})$ is also EFable.
We observe that in allocation $\mathbf{A}$, every house assigned to type-$x$ agents satisfies $d(h) \ge 1$, while every house assigned to type-$y$ agents satisfies $d(h) \le 1$.
Similarly, in allocation $\mathbf{A}'$, houses assigned to type-$x$ and type-$y$ agents satisfy $d(h) \ge 2$ and $d(h) < 2$, respectively.
\end{example}




Based on the above observations, we obtain the following lemma (with full proof provided in Appendix~\ref{appendix-section:two-types}).

\begin{lemma}
\label{lemma:threshold-d}
An allocation $\mathbf{A} = (A_x, A_y)$, with $A_x$ and $A_y$ denoting the sets of houses allocated to type-$x$ and type-$y$ agents, is EFable if and only if there exists a threshold $\theta$ satisfying:
\begin{equation}
\label{eq:threshold}
    \forall h \in A_x, \; d(h) \geq \theta 
\quad \text{and} \quad
\forall h \in A_y, \; d(h) \leq \theta.
\end{equation}

\end{lemma}

\begin{observation}
For any EF outcome, all agents of the same type receive an identical final utility, consisting of the values of their assigned houses and subsidies that ensure envy-freeness.
\end{observation}


As shown in Example~\ref{table:two-types}, we observe that $h_1$ is the most valuable house for type-$x$ agents. 
However, if $h_1$ is assigned to one agent of type-$x$, the other agent in $N_x$ can receive at most a house of value $3$, resulting in a total subsidy of at least $3$. 
Similarly, if $h_6$ is assigned to one agent of type-$y$, the other agent in $N_y$ can receive at most a house of value $2$, also leading to a total subsidy of at least $3$.
This example illustrates that allocating the $n$ houses with the highest values is not always optimal. Indeed, for a given threshold $\theta$, the total subsidy depends critically on the choice of the maximum-valued house assigned to type-$x$ agents and to type-$y$ agents, as formalized in the following lemma (see Appendix~\ref{appendix-section:two-types} for the full proof).

\begin{lemma}
\label{lemma:compute-subsidy}
Given an EFable allocation $\mathbf{A} = (A_x, A_y)$, the corresponding total subsidy $P$ can be represented by a quadruple $(\theta_x,\theta_y, M_x, M_y)$ as follows: 
 \begin{align*}
  \label{eq:total-subsidy-in-lemma}
   \scalebox{0.93}{$
     \begin{aligned} 
P=
    \begin{cases}
        n \cdot M_y +n_x \cdot \theta_{y}- \sum_{i \in N} v_i(A_i) & \text{if }  M_x -M_y < \theta_{y}; \\
        n \cdot M_x-n_y \cdot \theta_{x}-\sum_{i \in N} v_i(A_i) & \text{if } M_x -M_y > \theta_{x};\\
         n_x \cdot M_x+n_y \cdot M_y -\sum_{i \in N} v_i(A_i) & \text{otherwise.}    
    \end{cases}
   \end{aligned}
$}
   \end{align*}
where $M_x := \max_{h \in A_x} v_x(h)$,
$M_y := \max_{h \in A_y} v_y(h)$, and
$\theta_y := \max_{j \in N_y} d(A_j) \in \mathcal{D}$,
$\theta_x := \min_{i \in N_x} d(A_i) \in \mathcal{D}$ with $\theta_x \ge \theta_y$.

\end{lemma}

It follows that, to find an allocation with minimum subsidy, it suffices to determine thresholds $(\theta_x,\theta_y)$ with $\theta_x \ge \theta_y$ such that the allocation satisfies $\theta = \theta_x$, which is EFable by Lemma~\ref{lemma:threshold-d}.
For fixed $M_x$ and $M_y$ satisfying the threshold condition, the minimum total subsidy is achieved by maximizing $\sum_{i \in N} v_i(A_i)$ over all feasible allocations.
Hence, the subsidy is minimized when the values of houses assigned to type-$x$ (resp.\ type-$y$) agents are as close as possible to $M_x$ (resp.\ $M_y$), which coincides with the total subsidy expression in Lemma~\ref{lemma:compute-subsidy}.
Based on this analysis, we develop Algorithm~\ref{Alg-type-min-subsidy}.

\begin{algorithm}[t]
\caption{Minimum Subsidy for Two Types of Agents}
\label{Alg-type-min-subsidy}
\textbf{Input:} An instance with two types of agents $\mathcal{I}=(N_x \cup N_y,H,(v_x,v_y))$. \\ 
\textbf{Output:} An EFable allocation $\mathbf{A}$ and total subsidy $P$. 

\begin{algorithmic}[1]
\STATE Define  $\mathcal{D} := \{d(h_1), \ldots, d(h_m)\}$, $\mathcal{V}_x:=\{v_x(h) \mid h \in H\}$, $\mathcal{V}_y:=\{v_y(h) \mid h \in H\}$, $P \gets \infty$, $\mathbf{A} \gets \emptyset$;
    \FORALL{$\theta_x \in \mathcal{D}$}
    \FORALL{$\theta_y \in \mathcal{D}, \theta_y \le \theta_x$}
        \FORALL{$M_x \in \mathcal{V}_x$}
            \FORALL{$M_y \in \mathcal{V}_y$}
                \STATE $T_x:= \{h \in H \mid d(h) \geq \theta_x \land v_x(h) \le M_x\}$; 
                \STATE $T_y:=\{h \in H \mid d(h) \leq \theta_y \land v_y(h) \le M_y\}$;
                \STATE Construct a bipartite graph $G=(U \cup V, E)$ according to $T_x$ and $T_y$;  
                 \IF{there exists a maximum-weight perfect matching $\mathcal{M}$ in $G$}
                  \STATE Let $\mathbf{A}$ be the allocation induced by $\mathcal{M}$;   
                    \STATE Compute $P$; if it is smaller than the current best, update $P$;
                \ENDIF
            \ENDFOR
       \ENDFOR
   \ENDFOR
\ENDFOR
    \STATE \textbf{Return} An EFable allocation $\mathbf{A}$ and total subsidy $P$.
\end{algorithmic}
\end{algorithm}

\paragraph{The Algorithm.} 
The algorithm begins by constructing the following three finite sets:
\[
\mathcal{D} := \{d(h_1), d(h_2), \ldots, d(h_m)\},
\]
and   
\begin{align*}
\mathcal{V}_x:=\{v_x(h) \mid h \in H\},\  \mathcal{V}_y:=\{v_y(h) \mid h \in H\},
\end{align*}
which form the basis of potential optimal allocation.
Then we employ a quadruple nested loop to iterate through every possible parameter combination 
\begin{align*}
(\theta_x,\theta_y, M_x, M_y) \in \mathcal{D} \times \mathcal{D} \times \mathcal{V}_x \times \mathcal{V}_y. 
\end{align*}
Within each iteration, we first define the candidate sets of houses for each type of agents: 
\begin{align*}
T_x:= \{h \in H \mid d(h) \geq \theta_x \land v_x(h) \le M_x\}, \\
T_y:=\{h \in H \mid d(h) \leq \theta_y \land v_y(h) \le M_y\}.
\end{align*}
Then the algorithm finds the best possible allocation by solving a maximum-weight perfect matching problem. 
Construct a bipartite graph $G=(U \cup V, E)$, where $U$ denotes the set of agents; $V := T_x \cup T_y$ is the set of house candidates; Edges $e\in E$ and weights are defined as: For each agent $i \in N_x$ and house $h \in T_x$, add edge $(i, h)$ with weight $v_x(h)$; For each agent $j \in N_y$ and house $h \in T_y$, add edge $(j, h)$ with weight $v_y(h)$.
We then try to find a maximum-weight perfect matching $\mathcal{M}$ in $G$.
If the perfect matching exists, let $\mathbf{A}$ be the allocation induced by $\mathcal{M}$, and compute the corresponding total subsidy $P$. 
Then update the total subsidy $P$ if the total subsidy is smaller than the previous one.
After all parameter quadruplets have been tested, the algorithm returns the optimal allocation $\mathbf{A}$ and its corresponding total subsidy $P$.

We summarize the details of the algorithm in Algorithm~\ref{Alg-type-min-subsidy}, and obtain the following main theorem (with full proof provided in Appendix~\ref{appendix-section:two-types}).

\begin{theorem} 
\label{theorem:type-min-subsidy-EF}
For any instance $\mathcal{I} = (N_x \cup N_y, H, (v_x, v_y))$ with two types of agents, Algorithm~\ref{Alg-type-min-subsidy} runs in polynomial time and returns an EFable allocation $\mathbf{A}$ along with a subsidy $P$ minimizing the total subsidy among all EFable allocations.
\end{theorem}

\section{Conclusion and Open Problems}
We focus on the problem of achieving envy-freeness in house allocation via monetary subsidies, with the goal of minimizing the total subsidy required. 
We show that finding an allocation with minimum total subsidy is NP-hard even for binary instances. 
Consequently, we present polynomial time algorithms for instances with a bounded number of agent types, as well as for two types of agents with general utilities.
Several directions remain open for future research. 
One natural direction is to extend our results to general utilities with more than two agent types.
Our approach for instances with two types of agents relies on pairwise threshold structures, which are insufficient to characterize EFable allocations for three or more types.
Other interesting directions include considering additional fairness notions with subsidies, such as proportionality or maximin share guarantees.


%
%
%
 \bibliographystyle{plain}
 \bibliography{HAP}
\clearpage
\appendix

\section{Missing Proof of Theorem~\ref{theorem:reduction}}
\label{appendix:Missing-Proof-of-NP-hard}
\begin{proof}[\textbf{Proof of Theorem~\ref{theorem:reduction}}]

We will argue that there exists a subset $T \subseteq [m]$ with $|T| = k$ and $\left| \bigcup_{t \in T} S_t \right| \le q$ if and only if there exists an envy-free outcome with total subsidy at most $k+q$ in the corresponding binary instance.

$(\Rightarrow)$ Let $T \subseteq [m]$ be a feasible solution to the decision version of \problem{Minimum $k$-Union}. 
We construct an envy-free outcome $(\mathbf{A,P})$ such that $\sum_{i \in N}p_i \leq k+q$ as follows:
\begin{itemize}
    \item Assign to each agent $a^*_e \in N_1$ the specific house $h_{a^*_e}$.
    \item Assign to each agent $a_t$ the house $h_t$ for all $t \in [m]$.
    \item Assign to each agent $a_t$ with $t \in [m+1, m+k]$ exactly one distinct house from $\{h^*_j: j\in T\}$.
\end{itemize}

Next, we construct a subsidy vector $\mathbf{P}$ as follows:
\begin{equation*}
    p_i:=
    \begin{cases}
        1 & \text{if agent } i \in N_2 \text{ and } i \text{ gets } h^*_j \in T; \\
        1 & \text{if agent } i = a^*_e \text{ for some } e \in \bigcup_{j \in T} S_j;\\
        0 & \text{otherwise.}
    \end{cases}
\end{equation*}
Since $|\bigcup_{j \in T} S_j|  \leq q$ and  $k$ agents get $h^*_j \in T$, the total subsidy is $\sum_{i \in N}p_i \leq k+q.$

We now show that the outcome  $(\mathbf{A,P})$ is envy-free. 


Each agent $i \in N_1$ believes that her own allocated house has a value of $1$ and the houses allocated to any other agent $i' \in N_1 \setminus \{i\}$ have a value of $0$.
Therefore, agent $i$ does not envy any other agent $i'$, even if agent $i'$ receives a subsidy of $1$.
For each agent $i \in N_1$, she does not envy any agent $i' \in N_2$ receiving a subsidy of $1$, provided that agent $i$ considers the house allocated to agent $i'$ has a value of $0$.
Thus, only when agent $i$ considers the house in $H_2$ and the corresponding element $e$ is covered by some $S_j$ in $T$, with $S_j$ being allocated, should they receive a subsidy of $1$ to eliminate envy. 
Finally, every agent in $\{ a^*_e \mid e \in \bigcup_{j \in T} S_j \}$ attains a total utility of $2$ after receiving the subsidy. 
Each such agent evaluates the combined value (house and subsidy) received by agents in $\{ a_t \mid t \in [m+1, m+k] \}$ as at most $2$, and considers the house and subsidy received by agents in $\{ a_t \mid t \in [m] \}$ to be $0$.
Therefore, each agent $i \in N_1$ does not envy any other agent.

Note that every agent $i$ in $\{ a_t \mid t \in [m] \}$ has a utility of $1$ for her assigned house. 
Then they will not envy any agent $i'$ in $N_1 \cup \{ a_t \mid t \in [m+1, m+k] \}$ who receives a subsidy of $1$, since agent $i$ values the houses allocated to agent $i'$ at $0$.  
Thus these agents do not envy any other agent.

Finally, every agent $i$ in $\{ a_t \mid t \in [m+1, m+k] \}$ has a utility of $1$ after receiving the subsidy. 
Each such agent $i$ believes that every agent in $\{ a_t \mid t \in [m] \}$ receives a house valued at $1$ (with zero subsidy), and that each agent in $N_1$ receives a house valued at $0$ (with at most a subsidy of $1$). 
Thus, these agents are envy-free with respect to all other agents.

Therefore, the outcome $(\mathbf{A}, \mathbf{P})$ is envy-free.

$(\Leftarrow)$ Let $(\mathbf{A,P})$ be an envy-free outcome with total subsidy at most $k+q$, that is, $\sum_{i \in N}p_i\leq k+q.$ 

We will construct $T$ such that $|T|=k$ and $|\bigcup_{j \in T} S_j| \leq q$.
Let the set $T \subseteq [m]$ be defined as $T:=\{j\in[m]: A_i=h_j^* \text{ for some } i\in N_2\}.$
We claim that $T$ satisfies $|T| \geq k$ and $|\bigcup_{j \in T} S_j| \leq q$.

We first establish several useful claims.
\begin{claim}
\label{claim:h_e-should-be-allocated}
    In allocation $\mathbf{A}$, whenever $h_{a^*_e}$ is allocated, it must be assigned to its corresponding agent $a^*_e \in N_1$.
\end{claim}
\begin{proof}
    If $h_{a^*_e}$ is allocated, then agent $a^*_e$ must also be matched; otherwise, $\mathbf{A}$ would violate the permutation maximality condition stated in Theorem~\ref{theorem:SAGT2019}, and thus would not be EFable. 
    Therefore, there must exist some $h^*_j$ in allocation $\mathbf{A}$ such that $e \in S_j$. 
    Since $h_{a^*_e}$ has value $0$ for any agent in $N \setminus \{a^*_e\}$, exchanging $h_{a^*_e}$ with $h^*_j$ results in a new allocation that does not increase the total subsidy.
\end{proof}

\begin{claim}
\label{claim:N_2-get-H_2-and-H_3}
    For the agents in $N_2$, some agents are assigned houses from $H_2$, while others are assigned houses from $H_3$.
\end{claim}
\begin{proof}
    Since $|N| = n + m + k > n + m = |H_1| + |H_2| = |H_1| + |H_3|$, at least $k$ houses from $H_3$ must be allocated, and at least $k$ houses from $H_2$ must also be allocated.
    By Claim~\ref{claim:h_e-should-be-allocated}, the agents in $N_2$ do not receive houses from $H_1$. 
    Moreover, since $|N_2| = m + k > m = |H_2| = |H_3|$, some agents in $N_2$ are allocated houses from $H_2$, while others are allocated houses from $H_3$.
\end{proof}
    
    Combining these two claims, we conclude that every unmatched agent in $N_2$ will receive a house from $H_2$.  
    Since the total number of agents is $|N|=n + m + k > n + m=|H_1|+|H_3|$, at least $k$ agents in $N_2$ cannot be matched. 
    It follows that $|T| \geq k$.

    It remains to show that  $|\bigcup_{j \in T} S_j| \leq q$.  
    Since at least $k$ agents in $N_2$ cannot be assigned any house in $H_3$, each of them requires a subsidy of at least $1$ in order to remain envy-free with respect to an agent assigned to a house in $H_3$.
    Consequently, the total subsidy for these $|N_2|=m+k$ agents is at least $k$. 
    From the previous analysis, each unmatched agent is assigned a house from $H_2$.  
    If $|\bigcup_{j \in T} S_j| > q$, then at least $q$ covered agents in $ N_1 $ evaluate the total utility of unmatched agents as $2$.  
    Therefore, these covered agents also require a subsidy of at least $1$ to avoid envy toward the unmatched agents.  
    Hence, the minimum total subsidy satisfies  
     \begin{equation*}
     \sum_{i \in N} p_i \geq k + |\bigcup_{j \in T} S_j| > k + q,
    \end{equation*}  
    which contradicts the assumption that $\sum_{i \in N} p_i \leq k + q$. 
    It follows that $|\bigcup_{j \in T} S_j| \leq q$, as claimed.
\end{proof}

\section{Approximation Hardness}
\label{Appendix:approximation-preserving}

\begin{theorem}
\label{theorem:approximation-hardness}
    If there exists a polynomial-time $\alpha$-approximation algorithm for \textnormal{\problem{Minimum Subsidy for Binary}}, then there is a polynomial-time $\alpha$-approximation algorithm for \textnormal{\problem{Minimum $k$-Union}}.
\end{theorem}
\begin{proof}
     Indeed, the only difference between the reduction in Theorem~\ref{theorem:reduction} for the \problem{Minimum Subsidy for Binary} problem and the \problem{Minimum $k$-Union} problem is that the solution for \problem{Minimum Subsidy for Binary} involves an additional factor of $k$. 
The key idea is to use a reduction as in Theorem~\ref{theorem:reduction} and to amplify the parameter $q$ in the \textsc{Minimum $k$-Union} instance by replicating elements in $U$, thereby minimizing the influence of $k$ on the subsidy bound.

Let $\mathcal{I}$ be an instance of \textsc{Minimum $k$-Union} with a finite set $U:=\{e_1,e_2, \cdots, e_n\}$ of elements, subsets $S_1, S_2, \ldots, S_m$ of $U$, and positive integers $q, k$ such that $q \le |U|$ and $k < m$; 
the goal is to determine whether there exists a subset $T \subseteq [m]$ such that $|T| = k$ and $\left| \bigcup_{t \in T} S_t \right| \le q$.
We construct an instance $\mathcal{I}'$ of \problem{Minimum Subsidy for Binary} as follows. 
Choose a large enough replication factor $r$, and define a new universe $U' = U \times \{1, 2, \dots, r\}$, where each element $e \in U$ is replicated $r$ times. 
For each set $S_t$, define a new set $S'_t = S_t \times \{1, 2, \dots, r\}$ for all $S_t \in \{S_1, S_2, \ldots, S_m\}$ , where each element $e \in S_t$ is replicated $r$ times. 
The house instance $\mathcal{I}'$ is then formed based on these sets, following the same reduction as in Theorem~\ref{theorem:reduction}. 
The size of $\mathcal{I}'$ is $n' = |U'| = r \cdot n $, which is polynomial in $n$.

From the reduction, we have that for any subset $T \subseteq [m]$ with $|T| = k$, the union size in $\mathcal{I}$ is $|\bigcup_{t \in T} S_t| = q$ if and only if the total subsidy in $\mathcal{I}'$ is $k + r \cdot q$. 
Thus, the minimum subsidy for $\mathcal{I}'$ is $OPT_{MSB} = k + r \cdot OPT_{MU}$, where $OPT_{MU}$ is the optimal value for the instance $\mathcal{I}$.

Now, suppose we have an $\alpha$-approximation algorithm for \textsc{Minimum Subsidy for Binary}. Applied to $\mathcal{I}'$, it returns a solution with subsidy $P$ such that
\begin{equation*}
    P \le \alpha \cdot OPT_{MSB} = \alpha \cdot (k + r \cdot OPT_{MU}).
\end{equation*}
From the reduction, this solution corresponds to a subset $T \subseteq [m]$ with $|T|=k$ and union size $q$ in $\mathcal{I}$ satisfying:
\begin{equation*}
  P= k + r \cdot q. 
\end{equation*}
Thus,
\begin{align*}
    q = \frac{P - k}{r} &\le \frac{\alpha \cdot (k + r \cdot OPT_{MU}) - k}{r} \\
    &= \alpha \cdot \left(\frac{k}{r} + OPT_{MU}\right) - \frac{k}{r}.
\end{align*}
Setting $r>>k$, we have $\frac{k}{r} \to 0$. 
Therefore, $ q \le \alpha \cdot OPT_{MU}$ holds whenever $r$ is sufficiently large.
This means that the solution for $\mathcal{I}$ achieves a union size at most $\alpha \cdot OPT_{MU}$, which implies a polynomial-time $\alpha$-approximation algorithm for \textsc{Minimum $k$-Union}.

Therefore, if \textsc{Minimum Subsidy for Binary} has an $\alpha$-approximation algorithm, then \textsc{Minimum $k$-Union} also has an $\alpha$-approximation algorithm.   
\end{proof}

Combining Theorem~\ref{theorem:approximation-hardness} with the hardness of $O(n^{1/4})$ for \textsc{Minimum $k$-Union} under ``Dense versus Random'' conjecture for DkS to hypergraphs~\cite{chlamtavc2017minimizing}, we can deduce that \textsc{Minimum Subsidy for Binary} is hard to approximate to within a factor of $\Omega(n^{1/4})$.

\section{Missing Proofs in Section~\ref{section:k-types}}
\label{appendix-section:k-types}

\begin{proof}[\textbf{Proof of Lemma~\ref{lemma:min-subsidy-conditions}}] 

Lemma~\ref{subsidy-is-0} implies that an EFable allocation with minimum total subsidy corresponds to a maximum cardinality matching between $N$ and $H$, denoted by $\mathcal{M}$.

Since the unique partition $(N_S, N_L, H_S, H_L)$ satisfies conditions (a), (b), and (c) of Theorem~\ref{theorem-EFM-partition}, there exists an $H_S$-saturating matching in $G_{\mathcal{I}}[N_S, H_S]$ and an $N_L$-saturating matching in $G_{\mathcal{I}}[N_L, H_L]$.
It follows that the size of any maximum cardinality matching $\mathcal{M}$ in $G_{\mathcal{I}}$ must be $|H_S| + |N_L|$.
Thus, the size of $\mathcal{M}$ is $|H_S| + |N_L|$.

\textbf{Part (1).} 
Since there are no edges between $N_S$ and $H_L$ by condition (a) of Theorem~\ref{theorem-EFM-partition}, the neighborhood of $N_S$ is contained in $H_S$.
Suppose that a maximum-cardinality matching admits a subset $N_S$ of agents with neighbor set $H_S$ such that $|H_S| \ge |N_S|$.
Then the size of the maximum-cardinality matching is 
\begin{equation*}
|H_S| + |N_L| \geq |N_S| + |N_L| = |N|, 
\end{equation*}
which implies that every agent can be matched and no subsidies are required.

Next, we prove that in any EFable allocation induced by a maximum cardinality matching, all houses in $H_S$ must be allocated to agents in $N_S$.
Since this partition satisfies property (b) of Theorem~\ref{theorem-EFM-partition}, every house in $H_S$ must be matched in any maximum cardinality matching.
Now suppose that there exists a house $h \in H_S$ that is allocated to an agent $j \notin N_S$ in a maximum-cardinality matching $\mathcal{M}'$.
Therefore, only $|H_S|-1$ houses can be assigned to agents in $N_S$.
Since there are no edges between $N_S$ and $H_L$ by property~(a) of Theorem~\ref{theorem-EFM-partition}, agents in $N_S$ regard all houses in $H_L$ as having value $0$.
As a result, at least $|N_S|-(|H_S|-1)$ agents in $N_S$ cannot be assigned to any house of value $1$.
Consequently, agents in $N_S$ can be matched to at most $|H_S|-1$ houses, while at most $|N_L|$ agents in $N_L$ can be matched.
Therefore, the matching size satisfies 
\begin{equation*}
    |\mathcal{M'}| \leq |N_L|+|H_S|-1 < |N_L|+|H_S|=|\mathcal{M}|,
\end{equation*}
which contradicts the maximality of $\mathcal{M'}$.
Therefore, all houses in $H_S$ must be allocated to agents in $N_S$.

Since all maximum cardinality matchings match the same number of agents, and condition~(c) of Theorem~\ref{theorem-EFM-partition} guarantees that $G_{\mathcal{I}}[N_L,H_L]$ admits an $N_L$-saturating matching, every agent in $N_L$ must be matched in any maximum cardinality matching.
Hence, any unmatched agent in $\mathcal{M}$ can only belong to $N_S$.
Moreover, by condition~(b) of Theorem~\ref{theorem-EFM-partition}, in the subgraph $G_{\mathcal{I}}[N_S,H_S]$, at least $|H_S|$ agents in $N_S$ must be matched.
Let $t := |N_S| - |H_S|$ denote the number of unmatched agents, who can only be assigned houses in $H_L$ that they do not like in $\mathbf{A^*}$.
Since the total subsidy is nonzero, at least one agent must receive a subsidy of $1$. 
Thus, each unmatched agent in $\mathcal{M}$, that is, each of the $t$ unmatched agents, must receive a subsidy of $1$ to maintain envy-freeness.

\textbf{Part (2).} Since the subgraph $G_{\mathcal{I}}[N_L, H_L]$ admits an $N_L$-saturating matching by condition~(c) in Theorem~\ref{theorem-EFM-partition}, it follows that all agents in $N_L$ must be matched in any maximum cardinality matching.
Since $\mathbf{A^*}$ is induced by a maximum cardinality matching, all agents in $N_L$ must be matched.


\textbf{Part (3).} 
Let $N_0$ denote the set of agents unmatched in $\mathcal{M}$. 
Since all agents in $N_L$ must be matched by Part~(2), it follows that $N_0 \subseteq N_S$.
For agents in $N_S \setminus N_0$ who are matched, the fact that there are no edges between $N_S$ and $H_L$ implies that they do not envy any agent in $N_L$ or in $N_0$ who receives a subsidy of $1$. 
Indeed, they have a value of $0$ to the houses allocated to those agents, whereas their own allocated house has value $1$.
By Lemma~\ref{subsidy-is-1}, for any EFable allocation, each agent requires at most a subsidy of $1$. 
Therefore, the agents in $N_S \setminus N_0$ are strictly envy-free without any subsidy.
Thus, in any minimum total subsidy vector, in addition to allocating a subsidy of $1$ to each agent in $N_0$ by Part~(1), there exists a subset $\hat{N}_L \subseteq N_L$ such that each agent in $\hat{N}_L$ receives a subsidy of at most $1$, by Lemma~\ref{subsidy-is-1}.

\textbf{Part (4).} We prove this statement by contradiction. 
We first prove that all agents of the same type must belong either entirely to $N_S$ or entirely to $N_L$.
Assume that there exist two agents $i,j \in N$ of the same type such that $i \in N_S$ and $j \in N_L$. 
Since the subgraph $G_{\mathcal{I}}[N_L, H_L]$ admits an $N_L$-saturating matching, every agent $j \in N_L$ must be adjacent to some house $h \in H_L$.  
Because agents $i$ and $j$ are of the same type, agent $i$ must also be adjacent to house $h$, which contradicts that there are no edges between $N_S$ and $H_L$. 

Next we prove that all agents of the same type must belong either entirely to $\hat{N}_{L}$ or entirely to $N_L$.
For any agent $i \in \hat{N}_L$ who receives a subsidy of $1$, the agent is matched by Part~(2), and thus obtains a total utility of $2$.  
If there exists an agent $j \in N_{L} \setminus \hat{N}_{L}$ of the same type who receives no subsidy, then the total utility of agent $j$ at most $1$.  
In this case, agent $j$ would envy agent $i$, which contradicts the assumption that the minimum total subsidy outcome is envy-free.
\end{proof}

\paragraph{The Detailed Algorithm.} Our algorithm begins with a given partition $(N_{S}, \hat{N}_{L}, N_{L}\setminus\hat{N}_{L})$  and the new bipartite graph $G^*=(U \cup V,E)$. 
It proceeds in two phases: first updating $G^*$, and then computing a maximum matching in the updated graph $G^*$.
We remove all houses in $H_S$ and all agents in $N_S$ from $G^*$.
We update the graph $G^*$ by deleting the agent type vertices corresponding to $N_S$ in $U$ and the house type vertices corresponding to $H_S$ in $V$.
Formally, let
\[
U \leftarrow U \setminus U(N_S), \qquad V \leftarrow V \setminus V(H_S),
\]
where $U(N_S) \subseteq U$ (resp.\ $V(H_S) \subseteq V$) denotes the set of agent type (resp. house type) vertices corresponding to agent types whose agents belong to $N_S$ (resp. house types whose houses belong to $H_S$) (this update is feasible; see Lemma~\ref{lemma-delete-H_S} for a formal proof).
We define the house set
\[H^* := \{h \in H \mid \exists\, i \in N_{L}\setminus\hat{N}_{L} \ \text{such that} \ v_i(h)=1\},\]
i.e., the set of houses valued at $1$ by some agent in $N_{L}\setminus\hat{N}_{L}$. 
We then delete all edges between $\hat{N}_{L}$ and $H^*$, i.e.,
\[
E \;\leftarrow\; E \setminus \{\, e(i,h) \mid i \in \hat{N}_{L},\ h \in H^* \,\}.
\]
Deleting edges from $\hat{N}_L$ to $H^*$ is reasonable (see Lemma~\ref{lemma-delete-E} for a formal proof).
Moreover, by deleting all edges between $\hat{N}_{L}$ and $H^*$, we ensure that agents in $\hat{N}_{L}$ cannot be assigned houses in $H^*$. 
This is necessary because, in any allocation satisfying condition~(2) of Lemma~\ref{lemma:min-subsidy-conditions}, all agents in $\hat{N}_{L}$ must be matched, which requires them to be assigned to houses to which they are adjacent in $G^*$.
We construct a new agent type vertex $u$ with capacity $t = |N_S| - |H_S|$, representing the $t$ unmatched agents in condition (1) of Lemma~\ref{lemma:min-subsidy-conditions}, and add edges between $u$ and all remaining houses not in $H^*$.
Since these agents have a value of $0$ for all houses in $N_L$, it does not matter which house they receive, but they must not be assigned any house in $H^*$, as this would affect the envy-freeness of agents in $N_L \setminus \hat{N}_L$ who receive no subsidy.
The updated graph $G^*$ satisfies the conditions of Lemma~\ref{lemma:min-subsidy-conditions} (we show in Lemma~\ref{lemma:update-G} that such operations do not change the feasibility of partitions).
We use a maximum matching to check whether there exists an allocation that assigns each agent exactly one house.
This corresponds to a many-to-many matching problem, also known as a capacitated $b$-matching problem\footnote{In a capacitated $b$-matching problem, each vertex $v$ is assigned a capacity $b(v)$, restricting the number of incident edges that can be selected, and the goal is to maximize the total weight or the cardinality of the matching.}.
If the size of the maximum matching $\mathcal{M}^*$ in the updated graph $G^*$ equals $t + |N_L|$, which is the total capacity of the agent nodes, then $G^*$ admits a feasible house allocation; that is, every agent can be assigned a house.
In this case, we return a total subsidy of $t + |\hat{N}_{L}|$ and declare the partition feasible.
Otherwise, we declare that no house allocation satisfying the conditions of Lemma~\ref{lemma:min-subsidy-conditions} exists for this partition and return ``No such partition''.







We summarize the procedure to determine whether a given partition is feasible in Algorithm~\ref{Alg-Partition}. 

Next, we establish the following lemma that will be used to prove the correctness of our algorithm.

\begin{lemma}
    \label{lemma-delete-H_S} 
   If any house in a house type belongs to $H_S$, then all houses of that house type are included in $H_S$.
\end{lemma}
\begin{proof}
   We prove the lemma by contradiction. 
    Suppose there exists a house $h \in H_S$ in house type $v$, and another house $h' \in v$ such that $h' \notin H_S$. 
    Since $h \in H_S$, by condition~(b) of Theorem~\ref{theorem-EFM-partition}, $h$ is adjacent to some agents in $N_S$. 
    On the other hand, since $h' \notin H_S$, by condition~(a) of Theorem~\ref{theorem-EFM-partition}, $h'$ has no edges to any agent in $N_S$. 
    This contradicts the definition of house type, in which all houses have identical neighbors among agent types.
    Thus, all houses in the same house type must either all belong to $H_S$ or all not belong to $H_S$. 
\end{proof}

Note that property~(4) of Lemma~\ref{lemma:min-subsidy-conditions} implies that all agents of the same type either all belong to $N_S$ or none belong to $N_S$. 
Combining Lemma~\ref{lemma-delete-H_S}, the algorithm can directly remove from $U$ and $V$ all agent type and house type vertices corresponding to $N_S$ and $H_S$, respectively.

\begin{lemma}
    \label{lemma-delete-E}
    If any house in a house type belongs to $H^*$, then all houses of that house type are included in $H^*$.
\end{lemma}
\begin{proof}
    We prove the lemma by contradiction. 
Suppose there exists a house $h \in H^*$ in house type $v$, and another house $h' \in v$ such that $h' \notin H^*$. 
By the definition of $H^*$, there exists an agent $i \in N_{L}\setminus\hat{N}_{L}$ with $v_i(h)=1$. 
However, $h' \notin H^*$ implies that there does not exist any agent $i \in N_{L}\setminus\hat{N}_{L}$ with $v_i(h')=1$, which contradicts the fact that $h$ and $h'$ belong to the same house type $v$.
\end{proof}

\begin{lemma}
\label{lemma:update-G}
    The update operations on $G^*$ do not change the feasibility of any given partition $(N_{S}, \hat{N}_{L}, N_{L}\setminus\hat{N}_{L})$.
\end{lemma}

\begin{proof}
We prove the lemma by showing that none of the update operations removes any edge or capacity that is required by a feasible allocation.
The condition~(1) stated in Lemma~\ref{lemma:min-subsidy-conditions} implies that the houses $H_{S}$ should all be allocated to $N_{S}$ in $\mathbf{A^*}$.
Therefore, removing the houses in $H_S$ and the agents in $N_S$ in Line~1 of Algorithm~\ref{Alg-Partition} does not affect the feasibility of the partition, since these assignments are fixed in any feasible allocation.

In the optimal allocation $\mathbf{A}^*$, conditions~(1) and~(3) imply that the $t$ unmatched agents together with those in $\hat{N}_{L}$ who receive a subsidy of $1$ cannot be assigned any house that is valued at $1$ by agents in $N_{L}\setminus\hat{N}_{L}$. 
Otherwise, an agent in $N_{L}\setminus\hat{N}_{L}$, who receives zero subsidy, would obtain a total utility of $1$, while an agent in $\hat{N}_{L}$ (or among the $t$ unmatched agents), who receives a subsidy of $1$ and is assigned a house of value $1$, would obtain a total utility of $2$.
This would cause envy, thereby contradicting the fact that the allocation $\mathbf{A}^*$ with minimum total subsidy is envy-free.
Therefore, deleting all edges between $\hat{N}_{L}$ and $H^*$ only removes assignments that cannot appear in any feasible EFable allocation satisfying Lemma~\ref{lemma:min-subsidy-conditions}.
Hence, the update operation in Line~3 of Algorithm~\ref{Alg-Partition} preserves the feasibility of the partition.

Since these $t$ agents have zero valuation for all remaining houses, we add edges between them and the remaining houses not in $H^*$ to ensure that they can be assigned a house later, while preventing them from being assigned any house in $H^*$.
This operation in Line 6 of Algorithm~\ref{Alg-Partition} does not affect the total subsidy, since the agents in $N_{L}\setminus\hat{N}_{L}$ also have zero valuation for all remaining houses not in $H^*$; therefore, it does not change the feasibility of the partition.
\end{proof}



\begin{algorithm}[t]
\caption{Checking Partition Feasibility}
\label{Alg-Partition}
\textbf{Input:} A partition $(N_{S}, \hat{N}_{L}, N_{L}\setminus\hat{N}_{L})$ and $G^*=(U \cup V,E)$.\\
\textbf{Output:} Total subsidy can be $t+|\hat{N}_{L}|$ or ``No such partition''. 

\begin{algorithmic}[1]
\STATE Update $U \leftarrow U \setminus U(N_S), V \leftarrow V \setminus V(H_S)$;
\STATE Define $H^* := \{h \in H \mid \exists\, i \in N_{L}\setminus\hat{N}_{L} \ \text{s.t.} \ v_i(h)=1\}$; 
\STATE Update $E \leftarrow E\setminus \{e(i,h):i \in \hat{N}_{L}, h\in H^{*}\}$;
\STATE Let $t :=|N_{S}|-|H_{S}|$;
\STATE Construct a new agent type vertex $u$ with capacity $t$; 
\STATE Update $E \leftarrow E \cup \{(u,h) \mid h \in H \setminus H^{*}\} $;
\STATE Compute a maximum matching $\mathcal{M^*}$ for updated $G^*$;
\IF{the size of maximum matching is $t + |{N}_{L}|$}
\STATE \textbf{Return:} Total subsidy can be $t+|\hat{N}_{L}|$.
\ELSE 
\STATE \textbf{Return:} ``No such partition''.
\ENDIF
\end{algorithmic}
\end{algorithm}

\begin{proof}[\textbf{Proof of Theorem~\ref{theorem:k-types}}]
    We first prove the correctness of our algorithm and then show that it runs in polynomial time. 
    It can be argued that when the partition coincides with the partition defined by the optimal allocation $\mathbf{A^*}$ with the minimum total subsidy, the allocation must exist after updating the graph $G^*$ (by Lemma~\ref{lemma:update-G}). 
    Consequently, we can construct the corresponding allocation $\mathbf{A^*}$ as follows: 
    first, assign houses in $H_S$ to agents in $N_S$ according to a maximum cardinality matching $\mathcal{M}$ computed on the original graph $G^*$ before any updates; 
    then, allocate houses in $H_L$ to agents in $N_L$ and the remaining unmatched agents according to the maximum matching $\mathcal{M^*}$ computed on the updated graph $G^*$.
   Therefore, it suffices to output the allocation requiring the minimum subsidy.

    
    Next, we show that the algorithm runs in $O((k \cdot 2^{2k})^{(1+o(1))}) + O(n + m)$ time.
    The running time of our algorithm consists of two parts. 
    A preprocessing step takes $O(n + m)$ time, in which the $n$ agents are aggregated into $k$ types and the $m$ houses are grouped into at most $2^k$ categories. 
    After preprocessing, the size of the compressed instance depends only on $k$.
   Finding the unique EFM partition of a bipartite graph can be done in $O(|E|\sqrt{|U|+|V|})$ time~\cite{uniqueAigner-HorevS22}, which in our setting is bounded by $O(k \cdot 2^{1.5k})$.
    Clearly, given the partition $(N_{S}, \hat{N}_{L}, N_{L}\setminus\hat{N}_{L})$, we can test whether an allocation meeting all the conditions in Lemma~\ref{lemma:min-subsidy-conditions}, as demonstrated in Algorithm~\ref{Alg-Partition}.
    In particular, constructing the bipartite graph in Lines~1--6 of Algorithm~\ref{Alg-Partition} takes $O(k \cdot 2^k)$ time. 
    Note that the maximum-cardinality capacitated $b$-matching problem can be reduced to a maximum flow problem in $O(k \cdot 2^k)$ time~\cite{ahuja1993network}.
    We compute a maximum matching in updated $G^*$ using the maximum flow almost-linear-time algorithm of Chen et al.~\cite{chen2025maximum}, which achieves a running time of $|E|^{1+o(1)}$ for a network with $|E|$ edges. 
    In the flow network constructed from $G^*$, the number of edges is $O(k \cdot 2^k)$, giving a running time of $O\big((k\cdot2^k)^{1+o(1)}\big).$
    Moreover, the total number of partitions is at most $2^{k}$ by Observation~\ref{observation:partitions-number}.  
    The overall time complexity is $O(2^{k} \cdot (k \cdot 2^{k})^{(1+o(1))}) = O((k \cdot 2^{2k})^{(1+o(1))})$.
    Therefore, the overall running time is $O\!\left((k \cdot 2^{2k})^{(1 + o(1))}\right)+ O(n + m)$. 
\end{proof}

\section{Missing Proofs in Section~\ref{section:two-types}}
\label{appendix-section:two-types}

\begin{proof}[\textbf{Proof of Lemma~\ref{lemma:threshold-d}}]
    We first prove the “only if” direction by contradiction. 
    Assume that an agent $i$ of type-$x$ is assigned house $h_i$ and an agent $j$ of type-$y$ is assigned house $h_j$ in EFable allocation $\mathbf{A}$, with $d(h_i) < d(h_j)$.  
    Then the weight of the cycle $(i,j,i)$ is 
    \begin{align*}
        w(i,j)+w(j,i)&=\left(v_i(h_j)-v_i(h_i) \right)+\left( v_j(h_i)-v_j(h_j)\right) \\
        &= \left(v_i(h_j)-v_j(h_j)\right)-\left(v_i(h_i)-v_j(h_i)\right) \\
        &=d(h_j)-d(h_i)>0,
    \end{align*}
     which contradicts that the allocation $\mathbf{A}$ is EFable.

Next we prove the “if” direction.
Assume the allocation $\mathbf{A}=(A_x,A_y)$ satisfies the threshold property~(\ref{eq:threshold}).
We show that this implies that $\mathbf{A}$ is EFable by proving that it satisfies the equivalent condition of maximizing utilitarian welfare across all reassignments of its houses.

The utilitarian welfare of the allocation $\mathbf{A}$ is given by:
\begin{equation*}
C(\mathbf{A}) = \sum_{h \in A_x} v_x(h) + \sum_{h \in A_y} v_y(h).
\end{equation*}
Since all agents of the same type are identical, the only valid permutations are those that involve swapping houses between an agent of type-$x$ and an agent of type-$y$. 
Therefore, it is sufficient to show that no single swap of a house from an agent of type-$x$ with a house from an agent of type-$y$ can increase the total utilitarian welfare.

Let \(h' \in A_x\) and \(h^* \in A_y\) be arbitrary houses.   
Consider a new allocation $\mathbf{A}'$ formed by swapping these two houses. 
The utilitarian welfare of $\mathbf{A}'$ is 
\begin{align*}
C(\mathbf{A}')&= \left( \sum_{h \in A_x, h \neq h'} v_x(h) + v_x(h^*) \right) \\
 &+  \left( \sum_{h \in A_y, h \neq h^*} v_y(h) + v_y(h') \right).
\end{align*}
The difference in utilitarian welfare is:
\begin{align*}
    C(\mathbf{A}) - C(\mathbf{A}') &= \left( v_x(h') + v_y(h^*) \right) - \left( v_x(h^*) + v_y(h') \right) \\
    &= \left( v_x(h') - v_y(h') \right) - \left(v_x(h^*) - v_y(h^*)\right) \\
    &= d(h') - d(h^*) \ge 0,
\end{align*}
where the last inequality follows from $d(h') \ge \theta \ge d(h^*)$.
This means that the allocation $\mathbf{A}$ maximizes the utilitarian welfare across all reassignments of its houses. 
Thus, the allocation $\mathbf{A}$ is EFable by Theorem~\ref{theorem:SAGT2019}.
\end{proof}

\begin{proof}[\textbf{Proof of Lemma~\ref{lemma:compute-subsidy}}]
    Given an EFable allocation $\mathbf{A}$, Lemma~\ref{lemma:threshold-d} implies that it must satisfy  $\theta_x \geq \theta_y$.
    The corresponding minimum total subsidy can be characterized in terms of $M_x$ and $M_y$.
    Assume that 
    \begin{equation*}
     i^{*} \! =\arg\max_{i \in N_x}v_x(A_i) \  \textnormal{ and } \  j^{*} \! =\arg\max_{j \in N_y}v_y(A_j).   
    \end{equation*}
    If $\ell(i^*) > 0$, then we have $\ell(j^*) \leq 0$; otherwise, the allocation is not EFable.
    Let $p_j=M_y-v_y(A_j)$ for all $j \in N_y$ and let $p_{i^*}=\ell(i^*)$ and $p_{i}=\ell(i^*) +(M_x-v_x(A_i))$ for all $i \in N_x$ resulting in a subsidy vector $\mathbf{P}$.
    Next we prove that the resulting allocation is EF.
    Now let $u_x$ and $u_y$ denote the final utilities for type-$x$ and type-$y$ agents, respectively, after subsidies are applied.
    For agents of type-$x$, all agents receive 
     \begin{equation*}
    u_x=v_i(A_i)+p_i=M_x+\ell(i^*).
     \end{equation*}
    For agents in type-$y$, all agents receive $u_y=M_y$.
   Then they do not envy any other agent of the same type.
   
    Since $\ell(i^*) > 0$ and the allocation $\mathbf{A}$ is EFable, there does not exist any positive cycle by condition (3) of Theorem~\ref{theorem:SAGT2019}.  
    By Theorem~\ref{theorem:min-subsidy-ell(i)}, setting $p_{i^*} = \ell(i^*)$
    eliminates the envy of type-$x$ agents towards type-$y$ agents.  
    At the same time, type-$y$ agents do not envy type-$x$ agents, which follows from the fact that no positive cycle exists in any EFable allocation.

Theorem~\ref{theorem:min-subsidy-ell(i)} implies that setting  $p_k^* = \ell(k), \  \forall k \in N$ yields a subsidy vector $\mathbf{P^*}$ that achieves the minimum total subsidy.
Since $\ell(i^*) > 0$, $\forall i \in N_x \setminus \{i^*\}$  we have 
\begin{equation*}
\ell(i) = \ell(i^*) + \big(M_x - v_x(A_i)\big) 
\end{equation*}
and  $\forall j \in N_y$, we have 
\begin{equation*}
\ell(j) = M_y - v_y(A_j), 
 \end{equation*}
which coincides with the subsidy vector constructed above.  
Since agent $i^*$ does not envy any agent of type-$x$, we have

\begin{align*}
\ell(i^*) &= \ell(i^*,j^*) = \max_{j \in N_y} \left( w(i^*, j) + w(j, j^*)\right) \\
          &=\max_{j \in N_y} \left( v_{x}(A_{j})-M_x+M_y-v_{y}(A_{j})\right) \\
          &=\max_{j \in N_y} \left( M_y-M_x+d(A_{j}) \right) \\
          &= M_y-M_x + \max_{j \in N_y}d(A_{j}) \\
          &= M_y-M_x + \theta_{y}.
\end{align*}
Then the total subsidy of $\mathbf{A}$ is 
\begin{align*}
    P &=\sum_{i \in N_x} \left( \ell(i^*) +\left( M_x-v_x(A_i)\right) \right)+\sum_{j \in N_y}\left( M_y-v_y(A_j)\right) \\
    &= n_x \cdot \left( \ell(i^*)+ M_x\right) + n_y \cdot M_y - \sum_{k \in N} v_k(A_k) \\
    &= n \cdot M_y +n_x \cdot \theta_{y}- C(\mathbf{A}).
\end{align*}
For the case of $\ell(j^*) > 0$, the argument is symmetric, since the roles of type-$x$ and type-$y$ are arbitrary.
 If $\ell(j^*) > 0$, then we have 
\begin{align*}
    \ell(j^*) &=\ell(j^*,i^*)= \max_{i \in N_x} \left( w(j^*, i) + w(i, i^*)\right) \\
    &= \max_{i \in N_x} \left(  v_{y}(A_{i})-M_y+M_x-v_{x}(A_{i})  \right) \\
    &= M_x-M_y+ \max_{i \in N_x} \left( -d(A_{i}) \right) \\
    &= M_x-M_y-\min_{i \in N_x}d(A_{i})\\
    &= M_x-M_y-\theta_{x}.
\end{align*}
Then the total subsidy of $\mathbf{A}$ is 
\begin{align*}
    P = n \cdot M_x-n_y \cdot \theta_{x}-C(\mathbf{A}).
\end{align*}    
Thus the total subsidy can be summarized as 
 \begin{align}
  \label{eq:total-subsidy}
   \scalebox{1.0}{$
     \begin{aligned} 
     P=
    \begin{cases}
        n \cdot M_y +n_x \cdot \theta_{y}- C(\mathbf{A}) & \text{if }  \ell(i^*) > 0; \\
        n \cdot M_x-n_y \cdot \theta_{x}-C(\mathbf{A}) & \text{if } \ell(j^*) > 0;\\
         n_x \cdot M_x+n_y \cdot M_y -C(\mathbf{A}) & \text{otherwise.}    
    \end{cases}
   \end{aligned}
$}
   \end{align}
where $\theta_{x}$ can serve as the threshold $\theta$ in inequality~\eqref{eq:threshold}, with $\theta_x \geq \theta_y$.
Note that the conditions $\ell(i^*) > 0$ and $\ell(j^*) > 0$ cannot hold simultaneously; otherwise, the allocation $\mathbf{A}$ would not be EFable.
Thus, the minimum subsidy corresponding to $\mathbf{A}$ can be computed using $M_x$, $M_y$, and a pair of thresholds $\theta_x, \theta_y \in \mathcal{D}$.
\end{proof}

\begin{proof}[\textbf{Proof of Theorem~\ref{theorem:type-min-subsidy-EF}}]

We first prove the correctness of our algorithm and then show that it runs in polynomial time. 
Let allocation $\mathbf{A}$ be the final allocation returned by Algorithm~\ref{Alg-type-min-subsidy}. 
This allocation must have been computed during a specific iteration of the algorithm's main loop, corresponding to a particular parameter quadruple $(\theta_x,\theta_y, M_x, M_y)$.
It follows that the allocation $\mathbf{A}$ satisfies the threshold property, and therefore it is EFable.
It can be argued that when the parameters $(\theta_x,\theta_y, M_x, M_y)$ coincide with those of an EFable allocation achieving the minimum total subsidy, the algorithm attains the overall minimum total subsidy by computing the maximum-weight perfect matching, which identifies the allocation that maximizes utilitarian welfare, as established in Lemma~\ref{lemma:compute-subsidy}.

Next we prove that Algorithm~\ref{Alg-type-min-subsidy} terminates in polynomial time.
The initial phase of the algorithm involves creating three discrete sets that define the parameter space for the search.
Computing the valuation differences $d(h)$ for all $h \in H$ and sorting them takes $O(m \log m)$ time.
Thus, the total time complexity of the precomputation phase is $O(m \log m)$.
The algorithm employs a quadruple nested loop that iterates through every combination of parameters from the precomputed sets.
Therefore, the total number of iterations performed by the main loop is bounded by $O(m^4)$.
For each parameter quadruple $(\theta_x,\theta_y, M_x, M_y)$, the algorithm takes $O(m)$ time in lines 6 - 7.
Constructing a bipartite graph with $(n+m)$ vertices and finding a maximum weight perfect matching takes $O((n+m)^3)$ by Hungarian algorithm~\cite{kuhn1955hungarian}.
After obtaining an allocation $\mathbf{A}$, calculating the utilities and the total subsidy takes $O(n)$ time.
The total time complexity of the algorithm is the product of the number of main loop iterations $O(m^4)$ and the complexity of a single iteration $O((n+m)^3)$.
The overall time complexity of Algorithm~\ref{Alg-type-min-subsidy} is $O(m^7)$. 
\end{proof}

\end{document}